\newcommand{\CLASSINPUTinnersidemargin}{0.85in}
\newcommand{\CLASSINPUTbottomtextmargin}{2.35cm}

\documentclass[journal,12pt,draftcls,onecolumn,a4paper]{IEEEtran} 

\usepackage{ctable}
\usepackage{graphicx}
\usepackage{color}
\usepackage{units}
\usepackage{array}
\usepackage{amsmath} 
\usepackage{amssymb,amsthm} 
\usepackage{caption}
\usepackage[justification=centering]{subfig}
\usepackage[numbers,sort&compress]{natbib}
\usepackage{breqn}
\usepackage{changepage}

\newcommand{\definedas}{\overset{\underset{\Delta}{}}{=}}
\newcommand{\req}{\overset{\underset{!}{}}{=}}
\newcommand{\E}{\text{E}}
\newcommand{\pr}{\text{Pr}}
\newcommand{\Sa}{\mathcal{S}_{\text{a}}}
\newcommand{\Sac}{\mathcal{S}_{\text{a}}^{\complement}}

\newcommand{\e}{\text{e}}

\newcommand{\dd}{\text{d}}
\newcommand{\U}{\text{U}}
\newcommand{\argorder}{\operatornamewithlimits{argorder}} 
\newcommand{\argmin}{\operatornamewithlimits{argmin}} 
 
\theoremstyle{plain}
\newtheorem{theorem}{Theorem}

\theoremstyle{remark}
\newtheorem{remark}{Remark}

\makeatletter
\newcommand{\thickhline}{%
    \noalign {\ifnum 0=`}\fi \hrule height 1pt
    \futurelet \reserved@a \@xhline
}
\newcolumntype{"}{@{\hskip\tabcolsep\vrule width 1pt\hskip\tabcolsep}}
\makeatother

\begin{document}
%
\title{Multi-user Scheduling Schemes for Simultaneous Wireless Information and Power Transfer Over Fading Channels\vspace{-0.3cm}}
\author{\IEEEauthorblockN{Rania Morsi, Diomidis S. Michalopoulos, and Robert Schober}
\vspace{-1.5cm}
\thanks{This paper was presented in part at the IEEE International Conference on Communications (ICC), Sydney 2014 \cite{Multiuser_scheduling_Morsi}.}
\thanks{Rania Morsi, Diomidis S. Michalopoulos, and Robert Schober are with the Institute of Digital Communications, Friedrich-Alexander-University Erlangen-N\"urnberg (FAU), Germany, (email: morsi@lnt.de, michalopoulos@lnt.de, schober@lnt.de).}}
\maketitle
\begin{abstract}
Radio frequency energy harvesting presents a viable solution for prolonging the lifetime of wireless communication devices. 
In this paper,  we study downlink multi-user scheduling for a time-slotted system with simultaneous wireless information and power transfer. In particular, in each time slot, a single user is scheduled to receive information, while the remaining users opportunistically harvest the ambient radio frequency energy. We devise novel online scheduling schemes in which the tradeoff between the users' ergodic rates and their average amount of harvested energy can be controlled. In particular, we modify the well-known maximum signal-to-noise ratio (SNR) and maximum normalized-SNR (N-SNR) schedulers by scheduling the user whose SNR/N-SNR has a certain ascending order (selection order) rather than the maximum one. We refer to these new schemes as order-based SNR/N-SNR scheduling and show that the lower the selection order, the higher the average amount of harvested energy in the system at the expense of a reduced ergodic sum rate. The order-based N-SNR scheduling scheme provides proportional fairness among the users in terms of both the ergodic achievable rate and the average harvested energy. Furthermore, we propose an order-based equal throughput (ET) fair scheduler, which schedules  the user having the minimum moving average throughput out of the users whose N-SNR orders fall into a given set of allowed orders. We show that this scheme provides the users with proportionally fair average harvested energies. In this context, we also derive feasibility conditions for achieving ET with the order-based ET scheduler. Using the theory of order statistics, the average per-user harvested energy and ergodic achievable rate of all proposed scheduling schemes are analyzed and obtained in closed form for independent and non-identically distributed Rayleigh, Ricean, Nakagami-$m$, and Weibull fading channels. Our closed-form analytical results are corroborated by simulations. 
\end{abstract}
\vspace{-0.5cm}

\IEEEpeerreviewmaketitle
\vspace{-0.2cm}

\section{Introduction}
Because of the tremendous increase of the number of battery-powered wireless communication devices over the past decade, the idea of prolonging their lifetime by allowing them to harvest energy from the environment has recently drawn significant research interest. Wireless energy harvesting (EH) is particularly important for energy-constrained wireless networks, such as sensor networks. For such networks, replacing the device batteries can be costly and inconvenient in difficult-to-access environments, or even infeasible for sensors embedded inside the human body. This creates the need for utilizing renewable energy sources. However, common renewable energy resources such as solar and wind energy are uncontrollable, weather dependent, and not available indoors. On the other hand, harvesting energy from ambient or even dedicated radio frequency (RF) signals is a viable solution for supplying energy wirelessly to low-power devices \cite{powercast}. RF EH can be realized by converting the RF signal to a direct current (DC) signal with a rectenna, which is an antenna integrated with a rectifier (e.g. diode) \cite{RFtoDC2008}. The DC energy can then be used to  power battery-free devices, such as passive RF-identification (RFID) tags or to trickle charge low-power devices, such as wearable medical sensors. 

The fact that RF signals can transport both information and energy motivates the integration of RF EH into wireless communication systems. To this end, simultaneous wireless information and power transfer (SWIPT) was proposed in \cite{Varshney2008,Shannon_meets_tesla_Grover2010}, where it was shown that there exists a fundamental tradeoff between information rate and power transfer. This tradeoff can be characterized by the boundary of the so-called rate-energy (R-E) region \cite{MIMO_Broadcasting_Zhang2011_Journal}. However, in \cite{Varshney2008,Shannon_meets_tesla_Grover2010}, receivers were assumed to be able to harvest energy from a signal that has already been used for information decoding (ID), which is not possible with current technology due to practical circuit limitations \cite{WIPT_Architecture_Rui_Zhang_2012}. In \cite{MIMO_Broadcasting_Zhang2011_Journal}, two practical receiver architectures were proposed, where the receiver either switches in time between EH and ID, or splits the received signal to use one portion for ID and the remaining portion for EH.

Recently, multi-user systems employing SWIPT were studied in an effort to make SWIPT suitable for application in practical networks. A SWIPT system comprising two users with multiple antennas was studied in \cite{MIMO_Broadcasting_Zhang2011_Journal} and \cite{EH_ID_Scenarios_Clerckx} for broadcast and interference channels, respectively, where optimal transmission strategies were derived. Furthermore, multi-user multiple-input single-output SWIPT systems were studied in \cite{Multiuser_MISO_beamforming2013}, where the authors derived the optimal beamforming design that maximizes the total energy harvested by the EH receivers under signal-to-interference-plus-noise ratio constraints at the ID receivers. In \cite{Multiuser_OFDM_Kwan_Schober_Journal}, a multi-user scheduling algorithm was designed for a downlink orthogonal frequency division multiplexing (OFDM) system, where the system energy efficiency in bit/Joule was maximized for a minimum required sum rate and a minimum required energy harvested by the users.  Moreover, the authors of \cite{Throughput_Maximization_for_WPCN_Zhang2013} considered a multi-user time-division multiple-access system with energy and information transfer in the downlink and the uplink, respectively. The optimal downlink and uplink time allocation was derived with the objective of achieving maximum sum throughput or equal throughput under a total time constraint. 

For information-only transfer systems, multi-user scheduling schemes that exploit the independent and time-varying multipath fading of the users' channels for creating multi-user diversity (MUD) have been extensively studied \cite{Performance_Analysis_MUD_Alouini_2004, Unified_Scheduling_approach}. With such schemes, the user having the most favorable channel conditions is opportunistically scheduled to transmit/receive over the \emph{entire} time slot. For example, in maximum-throughput scheduling (maximum signal-to-noise ratio (SNR) scheduling) \cite{Performance_Analysis_MUD_Alouini_2004}, the user having the maximum SNR is scheduled and thus the sum rate is maximized. However, users with poor channel conditions (high path loss) may be deprived from gaining access to the channel. To avoid this disadvantage, maximum \emph{normalized} SNR (N-SNR) scheduling \cite{Performance_Analysis_MUD_Alouini_2004} schedules the user having the maximum N-SNR (normalized to its own average value) and thus maximizes the users' rates while ensuring that the rate of each user is proportional to the user's channel quality, providing proportional fairness. Another approach to providing fairness is to guarantee equal throughput (ET) to all users by scheduling in each time slot the user having the minimum moving average throughput \cite{equal_throughput_2009}.
To the best of the authors' knowledge, multi-user scheduling schemes that exploit MUD and provide long-term fairness among users have not been studied in the context of SWIPT so far. Hence,  in this paper, we provide a novel framework for scheduling multiple users in the downlink of a SWIPT system by modifying the aforementioned scheduling schemes such that the tradeoff between information rate and harvested energy can be controlled. The main contributions of this paper can be summarized as follows:
\begin{itemize}
\item We propose order-based SNR and proportionally-fair order-based N-SNR scheduling schemes, which are parametrized by a specific selection order. In particular, scheduling is performed by first ordering the users' SNR/N-SNR ascendingly and then scheduling for information transfer the user whose SNR/N-SNR order is the same as the selection order. The remaining users opportunistically harvest the ambient RF energy. Thereby, a small selection order implies that the good states of the time-varying channel are utilized for EH rather than for ID. This leads to a larger average amount of harvested energy at the expense of a reduction in the ergodic sum rate. We show that with the choice of the selection order, the R-E tradeoff can be controlled.
Different from \cite{Multiuser_scheduling_Morsi}, in this paper, the unfair order-based SNR scheduling is additionally studied as it provides the maximum possible sum rate for the maximum selection order and the maximum possible total harvested energy for the minimum selection order. Thus, it provides the extreme points for the upper bound of the R-E region of any scheduling scheme.
\item We devise an order-based ET fair scheduling scheme, where in each time slot the user having the minimum moving average throughput is scheduled among a set of users whose N-SNR orders fall into a given set of allowed orders. We show that the lower the orders in this set, the larger the average amount of energy harvested by the users at the expense of a reduction in their ET. Furthermore, feasibility conditions required for the users to achieve ET with this scheme are derived.
\item The order-based SNR, N-SNR, and ET scheduling schemes are analyzed using the theory of order statistics. Different from \cite{Multiuser_scheduling_Morsi}, in this paper, we analyze the three proposed schemes for independent and non-identically distributed (i.n.d.) Nakagami-$m$ and Weibull fading channels, in addition to Ricean and Rayleigh fading channels. In particular, closed-form expressions for the per-user average harvested energy and the per-user ergodic achievable rate are provided for the three proposed scheduling schemes and the considered fading models. For Ricean fading, we use a tight approximation for the Marcum Q-function which effectively transforms a Ricean fading channel into an equivalent Weibull fading channel and facilitates obtaining closed-form performance results for Ricean fading.
\end{itemize}
The rest of this paper is organized as follows. Section \ref{s:system_model} presents the overall system model, the adopted EH receiver model, and the considered fading channel models. This section also introduces two baseline scheduling schemes, namely round-robin (RR) and conventional ET scheduling, and analyzes their per-user average harvested energies and ergodic rates. Sections \ref{s:order_based_SNR}, \ref{s:order_based_NSNR_scheme}, and \ref{s:controllable_ET_scheduling} introduce respectively the order-based SNR, the proportionally-fair order-based N-SNR, and the order-based ET scheduling algorithms, and provide analytical results and closed-form expressions for the corresponding per-user average harvested energies and ergodic achievable rates. Additionally, feasibility conditions that have to be satisfied for the order-based ET scheduling scheme to achieve ET for all users are given in Section \ref{s:controllable_ET_scheduling}. Numerical and simulation results are provided in Section \ref{s:simulation_results}. Finally, Section \ref{s:conclusion} concludes the paper.

\section{System Model and Baseline Scheduling Schemes}
\label{s:system_model}
In this section, we present the multi-user SWIPT system model and the adopted EH receiver model. We also briefly review the fading channel models used for evaluating the performance of the proposed scheduling schemes. Additionally, two baseline scheduling schemes are provided to have a basis for evaluation of the proposed schedulers. However, we first introduce some notations and special functions used throughout the paper.

\subsection{Notations and Special Functions}
We use the following notations and functions throughout the paper. $\E[\cdot]$ denotes expectation. $\Gamma(s,x)$ is the upper incomplete Gamma function defined as $\Gamma (s,x)= \int_x^\infty{t^{s-1} \e^{-t} \dd t}$. $\Gamma(s)$ is the Gamma function defined as $\Gamma (s)= \int_0^\infty{t^{s-1} \e^{-t} \dd t}$ for positive non-integer $s$ and $\Gamma(s)=(s-1)!$ for positive integer $s$. $I_0(\cdot)$ is the modified Bessel function of the $1^{\text{st}}$ kind and order zero. $Q_1(a,b)$ is the first-order Marcum Q-function defined as $Q_1(a,b)=\int_b^\infty{x  \e^{-\frac{(x^2+a^2)}{2}} I_{0}(ax) \dd x}$. $G_{p,q}^{m,n}\left[z\middle|\begin{matrix} a_1,\ldots,a_p\\b_1,\ldots,b_q\end{matrix}\right]$ is the Meijer G-function defined in \cite[eq.(5)]{MeijerG1990}. $\E_1(x)=\int_1^{\infty} \frac{\e^{-tx}}{t}\quad \dd t$ is the exponential integral function of the first order. $\psi(x)$ is the digamma (psi) function defined in \cite[p.793]{prudnikov1986integrals}. $\text{C}=0.5772156649\ldots=-\psi(1)$ is the Euler constant \cite[p.787]{prudnikov1986integrals}. ${}_pF_q\left(\begin{aligned}&a_1,\ldots,a_p\\ &b_1,\ldots,b_q\end{aligned};z\right)$ is the generalized hypergeometric function defined in \cite[p.788]{prudnikov1986integrals}. $|\cdot|$ denotes the cardinality of a set. Finally, $\definedas$ stands for ``is defined as" and $\req$ stands for ``is required to be".

\subsection{System Model}
We consider a time-slotted SWIPT system that consists of one access point (AP) with a fixed power supply and $N$ user terminals (UTs) which are powered by the energy harvested from the AP downlink RF signal. Both the AP and the UTs are equipped with a single antenna. The system is studied for downlink transmission, where it is assumed that the AP always has a dedicated packet to transmit for each user. In each time slot, the AP schedules one user for information transmission, while the remaining idle users opportunistically harvest energy from the received signal, as shown in Fig. \ref{fig:ID_EH_system}. We assume that the receivers of the UTs employ time-switching \cite{MIMO_Broadcasting_Zhang2011_Journal}, i.e., each UT may use the received signal for either ID or EH. To this end, the AP broadcasts at the beginning of each time slot the index of the UT which is scheduled for ID, which is determined using the scheduling schemes introduced in Sections \ref{s:order_based_SNR}-\ref{s:controllable_ET_scheduling}. Moreover, we assume that all UTs are of the same type and require the same type of data (e.g., calibration, synchronization, or query signals sent in the downlink of wireless sensor networks).
\begin{figure}
\centering
\includegraphics[width=0.48\textwidth]{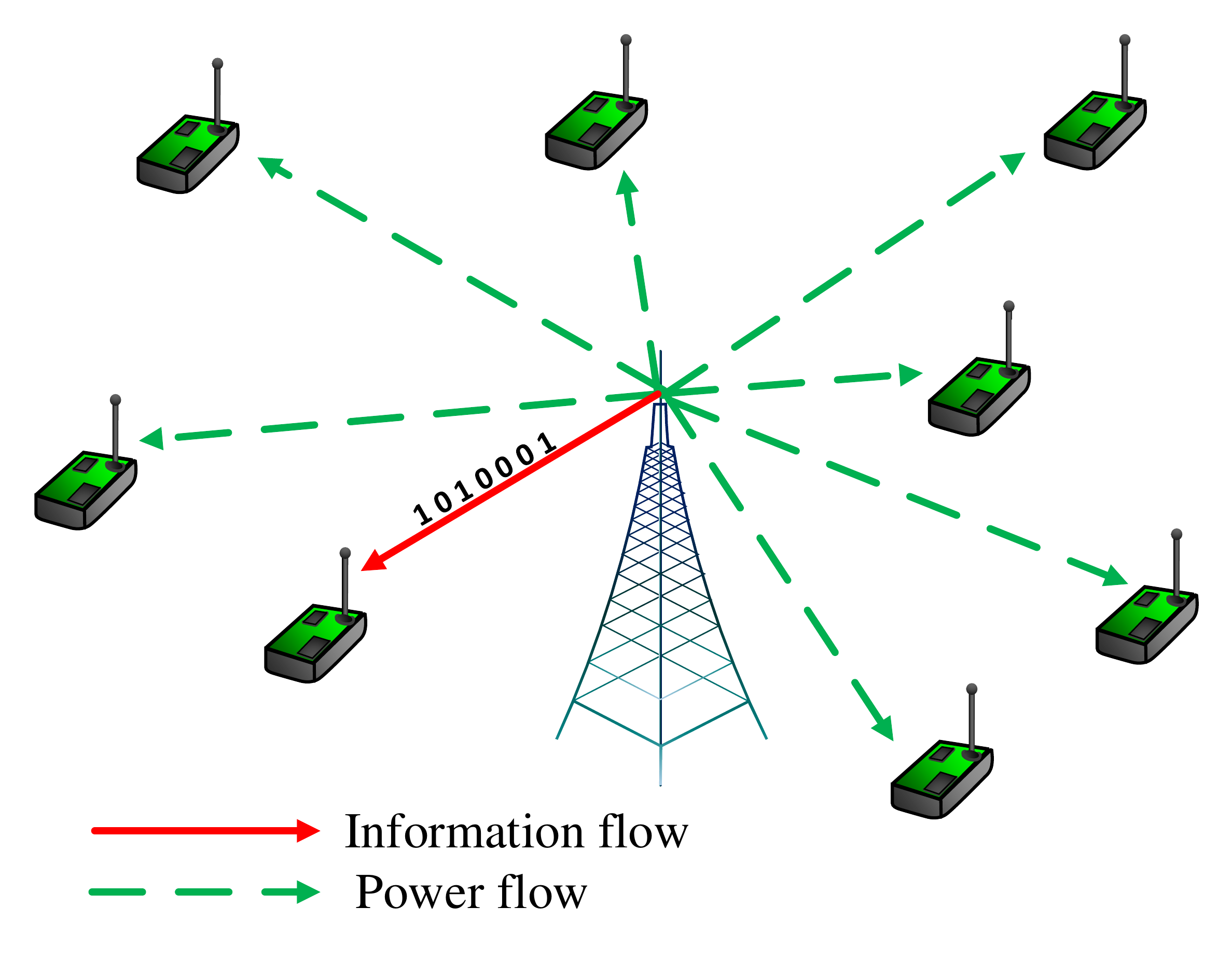}
\caption{Multi-user SWIPT system with time switching receivers. The scheduled user performs ID and the remaining users perform EH.}
\label{fig:ID_EH_system}
\end{figure}

At time slot $t$, the AP transmits an information signal to the scheduled user and the received signal at user $n$ is given by
\begin{equation}
r_n=\sqrt{P\,h_n}\e^{\text{j}\theta_n}x+z_{n},\quad \forall n\in\{1,\ldots,N\},
\label{eq:transmission}
\end{equation}
where $P$ is the constant transmit power of the AP\footnote{In this paper, we assume a fixed AP transmit power which is reasonable for simple wireless sensor networks. Power control provides an additional degree of freedom and thus may lead to a performance enhancement at the expense of a higher system complexity and a higher peak-to-average power ratio of the transmit signals. However, considering power control is out of the scope of this paper and left for future work.}, $x$ is the complex baseband information symbol assumed to be Gaussian distributed with average power normalized to 1, i.e., $\E[|x|^2]\!\!=\!\!1$, $\sqrt{h_n}$ and $\theta_n$ are respectively the amplitude and the phase of the fading coefficient of the channel from the AP to user $n$, and  $z_n$ is zero-mean complex additive white Gaussian noise with variance $\sigma^2$  impairing the received signal of user $n$.

\subsection{Energy Harvesting Receiver Model}
We adopt the EH receiver model described in \cite{WIPT_Architecture_Rui_Zhang_2012}. In this model, the average harvested power (or equivalently energy, for a unit-length time slot) is given by \cite[eq.(13)]{WIPT_Architecture_Rui_Zhang_2012}
\begin{equation}
EH=\eta h P,
\label{eq:EH}
\end{equation}
where $h$ is the channel power gain from the AP to the EH receiver and $\eta$ is the RF-to-DC conversion efficiency which ranges from 0 to 1. For current commercially available RF energy harvesters, $\eta$ can reach up to 0.7 \cite{powercast}.
\subsection{Channel Model} 
\label{ss:channel_model} 
The channels from the AP to the users are assumed to be block fading, i.e., the channel remains constant over one time slot, and changes independently from one slot to the next. The channel coefficients of the different user links are assumed to be independent and to have identical small-scale fading distributions but different path losses. This is a realistic assumption since the users may have different distances from the AP but they are in the same physical environment. We incorporate the effects of both small-scale fading and path loss into the channel power gain $h_n$ of user $n$, whose mean $\Omega_n$ is inversely proportional to the path loss (including the antenna gains).  We consider four fading channel models, namely Ricean, Nakagami-$m$, Weibull, and Rayleigh fading, i.e., $h_n$ follows respectively a non-central $\chi^2$, Gamma, Weibull, and exponential distribution. The corresponding probability density functions (pdfs) $f_{h_n}(x)$ and cumulative distribution functions (cdfs) $F_{h_n}(x)$ are given in Table \ref{tab:pdfs_cdfs_hn}. For Nakagami-$m$ fading, we assume that $m$ is an integer to facilitate the analysis.\\ 
Each considered fading model is suitable for modelling wireless power transfer in a different environment.  In particular, Nakagami-$m$ fading is a suitable model for indoor environments \cite{simon2005digital}. Weibull fading can effectively characterize the radio channel of narrow-band body area networks \cite{WBAN_Weibull_Justification}, which may solely depend on wireless power transfer when embedded inside human bodies. Rayleigh fading is a well-known special case of Nakagami-$m$ fading with $m=1$ and Weibull fading with $k=1$, cf. Table \ref{tab:pdfs_cdfs_hn}. Ricean fading may be considered as one of the most realistic channel models for short range wireless power transfer, as it includes a line of sight path. 

In order to gain analytical insight into the performance  of different scheduling schemes for Ricean fading,  we use a recently developed exponential approximation of the first-order Marcum-Q function \cite{EXP_Approx_Marcum_Q}
\begin{equation}
Q_1(a,b)\approx\e^{-\e^{\nu(a)}b^{\mu(a)}},
\label{eq:Exp_approx_Marcum_Q_function}
\end{equation}
where $\nu(a)$ and $\mu(a)$ are non-negative parameters given by\footnote{Note that for the case when $a$ is not $\ll 1$, the polynomials in (\ref{eq:mu}) and (\ref{eq:nu}) are slightly different from those in \cite[eq.(7)]{EXP_Approx_Marcum_Q}. We obtained the same optimal values for $\mu$ and $\nu$ as in \cite[Table I]{EXP_Approx_Marcum_Q}, but different fourth order coefficients of the least-square polynomial regression for $\nu(a)$ and $\mu(a)$. In fact, the polynomials in (\ref{eq:mu}) and (\ref{eq:nu}) yield a better approximation of the Marcum-Q function than the polynomials given in \cite[eq.(7)]{EXP_Approx_Marcum_Q}.}
\begin{equation}
\mu(a)= 
\begin{cases}
    2+\frac{9}{8(9\pi^2-80)a^4},& \text{if } a\ll 1\\
    2.1793-0.5916a+0.5895a^2-0.0909a^3+0.0053a^4, & \text{otherwise},
\end{cases}
\label{eq:mu}
\end{equation}
and
\begin{equation}
\nu(a)=
\begin{cases}
-\ln 2-\frac{a^2}{2}+\frac{45\pi^2+72\ln 2+36\text{\scriptsize C}-496}{64(9\pi^2-80)}a^4,& \text{if } a\ll 1\\
-0.8526+0.3504a-0.7529a^2+0.0858a^3-0.0045a^4, & \text{otherwise}.
\end{cases}
\label{eq:nu}
\end{equation}
The approximation in (\ref{eq:Exp_approx_Marcum_Q_function}) effectively transforms the Ricean fading channel into an equivalent Weibull fading channel, where $\lambda_n^k$ is replaced by $\beta_n$ and $k$ is replaced by $\mu'$, cf. Table \ref{tab:pdfs_cdfs_hn}.
\begin{table}[t]
  \caption{Pdf and cdf of the channel power gain $h_n$ of user $n$ for different fading models. $\Omega_n=\E[h_n]$ is the scale parameter of all fading distributions. $m$ and $k$ are the shape parameters of the Gamma and the Weibull distributions, respectively. $K$ is the ratio between the power in the direct path and that in the scattered paths for Ricean fading.}
\label{tab:pdfs_cdfs_hn}
\begin{tabular}{@{}llll@{}}  \toprule   
Fading Model & Pdf $f_{h_n}(x)$ & Cdf $F_{h_n}(x)$ &Parameters \\ \midrule \addlinespace[0.5em]
Nakagami-$m$ & $\frac{1}{\Gamma(m)}\lambda_n^m x^{m-1}\e^{-\lambda_n x}$
&\vspace{0.1cm}$\begin{aligned}&1-\frac{\Gamma\left(m,\lambda_n x\right)}{\Gamma(m)}\\=&1-\e^{-\lambda_n x}\sum\limits_{s=0}^{m-1}\frac{(\lambda_n x)^s}{s!}\end{aligned}$
&$\lambda_n=\frac{m}{\Omega_n}$ \\ \addlinespace[0.5em]
Weibull & $k\lambda_n^kx^{k-1}\e^{-(\lambda_nx)^k}$ & $1-\e^{-(\lambda_n x)^k}$ &\vspace{0.1cm}$\lambda_n=\frac{\Gamma\left(1+\frac{1}{k}\right)}{\Omega_n}$\\ \addlinespace[0.5em]
Ricean &$\frac{K+1}{\Omega_n} \e^{-K-\frac{(K+1)x}{\Omega_n}} I_0\left(2\sqrt{\frac{K(K+1)}{\Omega_n}x}\right)$ & $\begin{aligned}&1-Q_1\left(\sqrt{2K},\sqrt{\frac{2(K+1)x}{\Omega_n}}\right)\\&\approx 1-\e^{-\beta_nx^{\mu'}}\end{aligned}$ &\vspace{0.1cm}$\begin{aligned} \beta_n&= \e^{\nu(\sqrt{2K})}\sqrt{\frac{2(K+1)}{\Omega_n}}^{\mu(\sqrt{2K})}\\ \mu'&=\frac{\mu\left(\sqrt{2K}\right)}{2}\end{aligned}$\vspace{0.1cm}\\ \addlinespace[0.5em]
Rayleigh & $\lambda_n\e^{-\lambda_nx}$ & $1-\e^{-\lambda_n x}$ & $\lambda_n=\frac{1}{\Omega_n}$\\
\bottomrule
\end{tabular}
\end{table}
\begin{table}[!htp]
  \caption{Ergodic full-time-access achievable rate of user $n$ for different fading models \cite{Shannon_Capacity_fading_channels2005}. For Weibull fading, we use the definition $\lambda_n'= \frac{\Gamma\left(1+\frac{1}{k}\right)}{\bar{\gamma}_n}$ and parameters $a$ and $b$ are the smallest positive integers satisfying $\frac{b}{a}=k$. Function $\Delta(x,y)$ is defined as $\Delta(x,y)= \frac{y}{x},\frac{y+1}{x},\ldots,\frac{y+x-1}{x}$ \cite[p.792]{prudnikov1986integrals}.}
\label{tab:full_time_capacity}
\begin{tabular}{@{}ll@{}}  \toprule     
Fading Model & $\E[C_{\text{U}_{n,\text{f}}}]$ \\  \midrule \addlinespace[1em]
Nakagami-$m$ & $\frac{1}{\ln(2)\Gamma(m)}\left(\frac{m}{\bar{\gamma}_n}\right)^m G_{2,3}^{3,1}\left[\frac{m}{\bar{\gamma}_n}\middle|\begin{smallmatrix} &-m, &1-m \\ 0, &-m, & -m\end{smallmatrix}\right]$ \\  \addlinespace[1em] 
Weibull & $\frac{k\lambda_n^{'k}}{\ln(2)}\frac{\sqrt{a}b^{-1}}{(2\pi)^{\frac{a+2b-3}{2}}}G_{2b,a+2b}^{a+2b,b}\left[\frac{\lambda_n^{'ak}}{a^a}\middle|\begin{smallmatrix} \Delta(b,-k), &\Delta(b,1-k)\\ \Delta(a,0), &\Delta(b,-k), \Delta(b,-k)\end{smallmatrix}\right]$\\  \addlinespace[1em]  
Ricean & $\frac{(1+K)\e^{-K}}{\ln(2)\bar{\gamma}_n}\sum\limits_{i=0}^{\infty}\frac{1}{(i!)^2}\left[\frac{K(1+K)}{\bar{\gamma}_n}\right]^{i} G_{2,3}^{3,1}\left[\frac{(K+1)}{\bar{\gamma}_n}\middle|\begin{smallmatrix} &-1-i, &-i \\ 0, &-1-i, &-1-i\end{smallmatrix}\right]$\\ \addlinespace[1em]    
Rayleigh & $\frac{1}{\ln(2)}\e^{\frac{1}{\bar{\gamma}_n}}\,\E_1\left(\frac{1}{\bar{\gamma}_n}\right)$\\ 
\bottomrule
\end{tabular}
\end{table}
\subsection{Baseline Scheduling Schemes}
Next, we discuss two well-known schedulers, namely the round robin and conventional equal throughput schedulers, which will serve as baseline schemes for the newly proposed schedulers. To facilitate the comparison, we analyze the per-user ergodic rate and the average amount of harvested energy of the baseline schedulers.
\subsubsection{Round Robin (RR) Scheduling}
\label{ss:RR}
The RR scheduler grants the channels to the users in turn. Therefore, a user receives information with probability $\frac{1}{N}$, and  harvests energy with probability $1-\frac{1}{N}$. The AP only has to know the instantaneous channel of the scheduled user. Hence, user $n$ (denoted by $\U_n$) achieves an ergodic rate $\E[C_{\U_n}]$ that is $\frac{1}{N}$ times the ergodic rate $\E[C_{\text{U}_{n,\text{f}}}]$ that it would have achieved if it had full-time access to the channel. That is, 
\begin{equation}
\E[C_{\U_n}]\Big|_{\text{RR}}=\frac{1}{N} \E[C_{\text{U}_{n,\text{f}}}], 
\end{equation}
where $\E\left[C_{\text{U}_{n,\text{f}}}\right]=\int\limits_0^\infty{\log_2\left(1+\bar{\gamma}x\right)f_{h_n}(x)\,\, \dd x}$, with $\bar{\gamma}\definedas\frac{P}{\sigma^2}$, is derived in \cite{Shannon_Capacity_fading_channels2005} for the considered fading models and summarized here in Table \ref{tab:full_time_capacity}. The average harvested energy of user $n$ is $\E\left[EH_{\U_n}\right]=\left(1-\frac{1}{N}\right)\eta P \Omega_n$.
\subsubsection{Conventional Equal Throughput (ET) Scheduling}
\label{ss:ET}
Conventional ET scheduling is based on the idea that the long-term average throughput of all users can be made identical by scheduling in each time slot the user having the minimum moving average throughput \cite{equal_throughput_2009}. At time slot $t$, the conventional ET scheduler schedules for information transmission the user $n^*$ that satisfies 
 \begin{equation}
n^*=\argmin\limits_{n\in\{1,\ldots,N\}} r_n(t-1),
\label{eq:order_ET_conventional_rule}
\end{equation}
where $r_n(t-1)$ is the throughput of user $n$ averaged over previous time slots up to slot $t-1$. The throughput of the users is updated recursively as follows
\begin{equation}
r_n(t)= \left\{ \begin{array}{ll}
(1-\beta)r_n(t-1)+\beta C_n(t) &\mbox{if user $n$ is scheduled} \\
(1-\beta)r_n(t-1) &\mbox{otherwise}\\
\end{array}, \right.
\label{eq:throughput_update}
\end{equation}
where $h_n(t)$ and $C_n(t)=\log_2\left(1+\bar{\gamma}h_n(t)\right)$ are respectively the channel power gain and the achievable rate of user $n$ at time slot $t$. $\beta\in(0,1)$ is a smoothing factor which weighs new throughput values and $1-\beta$ discounts past values\footnote{$\beta$ is chosen to asymptotically vanish (e.g., $\beta=1/t$) in order to ensure that the moving average throughput $r_n(t)$ converges to its ensemble average $\E\left[C_{\U_n}\right]$ for the considered stationary fading process $h_n(t)$ (see \cite{Unified_Scheduling_approach} and the references therein).}. Using the conventional ET scheduler, the average rate of user $n$ has to fulfill
\begin{equation}
\E\left[C_{\text{U}_n}\right]=p_{n}\E\left[C_{\text{U}_{n,\text{f}}}\right]\req r,\quad\forall n=\{1,\ldots,N\},
\end{equation}
where $p_n$ is the channel access probability of user $n$, $\E\left[C_{\text{U}_{n,\text{f}}}\right]$ is the average full-time-access achievable rate of user $n$ provided in Table \ref{tab:full_time_capacity}, and $r$ is the equal throughput achieved by all users. Hence, the required access probability of user $n$ to achieve throughput $r$ is $p_n=\displaystyle r/\E\left[C_{\text{U}_{n,\text{f}}}\right]$. Since $\sum\limits_{n=1}^N p_n\!=\!1$ must hold, the ET that all users achieve with this scheme is 
\begin{equation}
E\left[C_{\text{U}}\right]\Big|_{\text{ET}}=\left({\sum\limits_{n=1}^{N}{\frac{1}{\E\left[C_{\text{U}_{n,\text{f}}}\right]}}}\right)^{-1},
\end{equation}
and the scheduling probability required for user $n$ reduces to $p_{n}\!=\!1\Big/\left({\sum\limits_{i=1}^{N}{\displaystyle  \frac{\E\left[C_{\text{U}_{n,\text{f}}}\right]}{\E\left[C_{\text{U}_{i,\text{f}}}\right]}}}\right)$. The average harvested energy of user $n$ is $\E\left[EH_{\U_n}\right]=\left(1-p_n\right)\eta P \Omega_n$.

We note that round robin and conventional ET scheduling are not biased towards power transfer nor towards information transfer. Hence, both schemes achieve only one feasible point in the R-E region of each user. In the following, we propose three scheduling schemes which enable the control of the R-E tradeoff for each user.

\section{Order-based Scheduling Schemes}
\subsection{Order-based SNR Scheduling}
\label{s:order_based_SNR}
By exploiting the knowledge of the users' channels at the access point for the selection process\footnote{The channel coefficients of the AP-UT links can be fed back by the users in the uplink in frequency division duplex (FDD) systems  or can be assumed available in time division duplex (TDD) systems due to channel reciprocity \cite{Unified_Scheduling_approach}.}, we propose an order-based SNR scheduling scheme, which schedules the user having the $j^{\text{th}}$ ascendingly-ordered SNR for receiving information, where order $j$  is a design parameter chosen from $\{1,\ldots,N\}$. If $j=N$, order-based SNR scheduling reduces to maximum-SNR scheduling. 
\subsubsection{Order-based SNR Scheduling Algorithm}
Since $P$ and $\sigma^2$ are identical for all users, ordering the SNRs in our model is equivalent to ordering the channel power gains. Hence, the order-based SNR scheduler selects for information transmission the user $n^*$ that satisfies 
 \begin{equation}
n^*=\argorder\limits_{n\in\{1,\ldots,N\}} h_n,
\label{eq:order_SNR_rule}
\end{equation}
where we define \textquotedblleft$\argorder$" as the argument of the $j^{\text{th}}$ ascending order. 
\subsubsection{Performance Analysis}
To analyze the per-user ergodic rate and average harvested energy for the order-based SNR scheduling scheme, the instantaneous channel power gains $h_n,\,\, n=1,\ldots,N,$ of all users are ascendingly ordered as $h_{(1)}\leq h_{(2)}\leq ...\leq h_{(N)}$, where $h_{(j)}$ is the $j^{\text{th}}$ smallest channel power gain. The pdf of $h_{(j)}$ for i.n.d. channels is given by (equivalent to \cite[eq.(1.1)]{order_statistics_ind_1994})
\begin{equation} 
f_{h_{(j)}}(x)=\sum\limits_{n=1}^N\sum\limits_{\mathcal{P}_n}\prod\limits_{l=1}^{j-1}{F_{h_{i_l}}(x)}\,f_{h_{n}}(x)\\
\prod\limits_{l=j}^{N-1}{\left(1-F_{h_{i_l}}(x)\right)},
\label{eq:pdf_order_statistics_ind}
\end{equation}
where $f_{h_{n}}(x)$ and $F_{h_{n}}(x)$ are respectively the pdf and the cdf of the channel power gain of user $n$ given in Table \ref{tab:pdfs_cdfs_hn} for the considered fading channels, $\sum_{\mathcal{P}_n}$ denotes the summation over all $\binom{N-1}{j-1}$ permutations ($i_1,\ldots,i_{N-1}$) of $(1,\ldots,n-1,n+1,\ldots,N)$ for which $i_1 <\ldots < i_{j-1}$ and  $i_j <\ldots < i_{N-1}$. For a given order $j$, the ergodic rate achieved by user $n$ is given by
\begin{equation}
\E\left[C_{j,\U_n}\right]=\int\limits_0^\infty{\log_2\left(1+\bar{\gamma}x\right)f_{h_{n}}(x)\sum\limits_{\mathcal{P}_n}\prod\limits_{l=1}^{j-1}{F_{h_{i_l}}(x)}\,\prod\limits_{l=j}^{N-1}{\left(1-F_{h_{i_l}}(x)\right)}\dd x},
\label{eq:Order_SNR_per_user_Ergodic_capacity}
\end{equation} 
and the average amount of energy harvested by user $n$ is given by
\begin{equation} %
\E\left[EH_{j,\U_n}\right]=\eta P\int\limits_0^\infty x f_{h_{n}}(x)\left(1-\sum\limits_{\mathcal{P}_n}\prod\limits_{l=1}^{j-1}{F_{h_{i_l}}(x)}
\prod\limits_{l=j}^{N-1}{\left(1-F_{h_{i_l}}(x)\right)}\right) \dd x,
\label{eq:Order_SNR_per_user_EH}
\end{equation}
since a user harvests energy when it is not scheduled (i.e., when the order of its SNR is not $j$).
\begin{table}[!t]
\caption{High-SNR approximation for the ergodic per-user rate over i.n.d. Nakagami-$m$, Weibull, and Ricean fading channels using order-based SNR scheduling. The exact ergodic per-user rate is shown for Rayleigh fading. Sets $\mathcal{U}_{n,r}$ and $\mathcal{S}_{m,r}$ are defined in (\ref{eq:set_Unr}) and (\ref{eq:set_Sm}), respectively.}
\label{tab:avg_Capacity_SNR}  
\begin{tabular}{@{}ll@{}}\toprule 
Fading Model & $\E\left[C_{j,\U_n}\right]$ \\ \midrule 
	 		Nakagami-$m$ & $\begin{aligned}\frac{\lambda_n^m}{\ln(2)\Gamma(m)}\sum\limits_{r=0}^{j-1}(-1)^r\sum\limits_{\mathcal{U}_{n,r}}\sum\limits_{\mathcal{S}_{m,r}}\frac{\prod\limits_{t=1}^{N-j+r}\lambda_{u_t}^{s_t}}{\prod\limits_{t=1}^{N-j+r}s_t!} &\left(\lambda_n+\sum\limits_{t=1}^{N-j+r}\lambda_{u_t}\right)^{-\left(m+\sum\limits_{t=1}^{N-j+r} s_t\right)} \Gamma\left(m+\sum\limits_{t=1}^{N-j+r} s_t\right)\\[-1em] &\left(\psi\left(m+\sum\limits_{t=1}^{N-j+r} s_t\right)+\ln\Bigg(\frac{\bar{\gamma}}{\lambda_n+\sum\limits_{t=1}^{N-j+r}\lambda_{u_t}}\Bigg)\right)\end{aligned}$
		\\ \addlinespace[1em]
 Weibull& $\begin{aligned}\frac{1}{\ln(2)}\lambda_n^k\sum\limits_{r=0}^{j-1}(-1)^r\sum\limits_{\mathcal{U}_{n,r}}\frac{1}{\left(\lambda_n^k+\sum\limits_{t=1}^{N-j+r}\lambda_{u_t}^k\right)}\left(\ln\left(\bar{\gamma}\right)-\frac{1}{k}\left(\ln\left(\lambda_n^k+\sum\limits_{t=1}^{N-j+r}\lambda_{u_t}^k\right)+\textrm{C}\right)\right)\end{aligned}$
\\ \addlinespace[1em]
 Ricean& $\begin{aligned}\frac{1}{\ln(2)}\beta_n\sum\limits_{r=0}^{j-1}(-1)^r\sum\limits_{\mathcal{U}_{n,r}}\frac{1}{\left(\beta_n+\sum\limits_{t=1}^{N-j+r}\beta_{u_t}\right)}\left(\ln\left(\bar{\gamma}\right)-\frac{1}{\mu'}\left(\ln\left(\beta_n+\sum\limits_{t=1}^{N-j+r}\beta_{u_t}\right)+\textrm{C}\right)\right)\end{aligned}$
\\ \addlinespace[1em]
 Rayleigh &	$\begin{aligned}\frac{1}{\ln(2)}\lambda_n\sum\limits_{r=0}^{j-1}(-1)^r\sum\limits_{\mathcal{U}_{n,r}}\frac{1}{\left(\lambda_n+\sum\limits_{t=1}^{N-j+r}\lambda_{u_t}\right)}\e^{\frac{1}{\bar{\gamma}}\left(\lambda_n+\sum\limits_{t=1}^{N-j+r}\lambda_{u_t}\right)}\E_1\left(\frac{1}{\bar{\gamma}}\left(\lambda_n+\sum\limits_{t=1}^{N-j+r}\lambda_{u_t}\right)\right)\end{aligned}$
\\		\bottomrule
		  \end{tabular}	
\end{table}
\begin{table}[!htbp]
  \caption{Average per-user harvested energy for i.n.d. fading channels using order-based SNR scheduling. Sets $\mathcal{U}_{n,r}$ and $\mathcal{S}_{m,r}$ are defined in (\ref{eq:set_Unr}) and (\ref{eq:set_Sm}), respectively.}
\label{tab:avg_EH_SNR}
\begin{tabular}{@{}ll@{}}\toprule 
Fading Model & $\E\left[EH_{j,\U_n}\right]$ \\ \midrule 
	 		Nakagami-$m$ & $\eta P\left(\Omega_n-\frac{\lambda_n^m}{\Gamma(m)}\sum\limits_{r=0}^{j-1}(-1)^r\sum\limits_{\mathcal{U}_{n,r}}\sum\limits_{\mathcal{S}_{m,r}}\frac{\prod\limits_{t=1}^{N-j+r}\lambda_{u_t}^{s_t}}{\prod\limits_{t=1}^{N-j+r}s_t!} \left(\lambda_n+\sum\limits_{t=1}^{N-j+r}\lambda_{u_t}\right)^{-\left(m+1+\sum\limits_{t=1}^{N-j+r} s_t\right)}\Gamma\left(m+1+\sum\limits_{t=1}^{N-j+r} s_t\right)\right)$
		\\ \addlinespace[0.5em]
 Weibull& $\eta P\left(\Omega_n-\lambda_n^k \Gamma\left(1+\frac{1}{k}\right) \sum\limits_{r=0}^{j-1}(-1)^r\sum\limits_{\mathcal{U}_{n,r}}\left(\lambda_n^k+\sum\limits_{t=1}^{N-j+r}\lambda_{u_t}^k\right)^{-\left(1+\frac{1}{k}\right)}\right)$
\\ \addlinespace[0.5em]
 Ricean& $\eta P\left(\Omega_n-\beta_n \Gamma\left(1+\frac{1}{\mu'}\right) \sum\limits_{r=0}^{j-1}(-1)^r\sum\limits_{\mathcal{U}_{n,r}}\left(\beta_n+\sum\limits_{t=1}^{N-j+r}\beta_{u_t}\right)^{-\left(1+\frac{1}{\mu'}\right)}\right)$
\\ \addlinespace[0.5em]
 Rayleigh &	$\eta P\left(\Omega_n-\lambda_n\sum\limits_{r=0}^{j-1}(-1)^r\sum\limits_{\mathcal{U}_{n,r}}\frac{1}{\left(\lambda_n+\sum\limits_{t=1}^{N-j+r}\lambda_{u_t}\right)^2}\right)$
\\		\bottomrule
		  \end{tabular}	
\end{table}

Closed-form expressions for the per-user average rate and harvested energy in (\ref{eq:Order_SNR_per_user_Ergodic_capacity}) and (\ref{eq:Order_SNR_per_user_EH}), respectively, are derived in Appendix \ref{app:proof_EH_C_SNR} for Nakagami-$m$ and Weibull fading. Since RF EH targets short-range application scenarios, we provide in Table \ref{tab:avg_Capacity_SNR} lower bounds for the ergodic rates which become tight at high SNR for Nakagami-$m$ and Weibull fading, cf. Appendix \ref{app:proof_EH_C_SNR}. Results for Ricean fading are deduced from those for Weibull fading by replacing $k$ by $\mu'$ and $\lambda_n^{k}$ by $\beta_n$, cf. Section \ref{ss:channel_model}. Therefore, for Ricean fading, the high-SNR rate approximation is not a lower bound since it is obtained using the Marcum-Q function approximation. The exact average per-user rate for Rayleigh fading is obtained by analytically solving (\ref{eq:Order_SNR_per_user_Ergodic_capacity}).  To obtain the exact ergodic rate for the other considered fading distributions, the integral in (\ref{eq:Order_SNR_per_user_Ergodic_capacity}) can be solved numerically. Table \ref{tab:avg_EH_SNR} shows the per-user average harvested energies for the considered fading channels, which are valid for all SNR regimes.

\subsection{Order-based Normalized-SNR (N-SNR) Scheduling}
\label{s:order_based_NSNR_scheme}
Since order-based SNR scheduling in (\ref{eq:order_SNR_rule}) may deprive some of the users from receiving information if the users have different average channel conditions, we propose a proportionally fair order-based \emph{normalized}-SNR scheduler which selects for information transmission the user having the $j^{\text{th}}$ ascendingly-ordered N-SNR. 
\subsubsection{Order-based N-SNR Scheduling Algorithm}
The order-based N-SNR scheme schedules for information transmission the user $n^*$ that satisfies 
\begin{equation}
n^*=\argorder\limits_{n\in\{1,\ldots,N\}} \frac{h_n}{\Omega_n}.
\label{eq:order_normalized_power_gain_rule}
\end{equation}
The normalization in (\ref{eq:order_normalized_power_gain_rule}) ensures that all users access the channel on average an equal number of times, and thus proportional fairness is ensured in terms of both the ergodic rate and the average amount of harvested energy. 
\subsubsection{Performance Analysis}
The random variables (RVs) to be ordered, $X_n \definedas \frac{h_n}{\Omega_n}$, have the same distribution as $h_n$ but with unit mean $\forall n\in\{1,\ldots,N\}$. Since all user channels are assumed to have the same shape parameter (i.e., the same Ricean factor $K$ for Ricean fading, the same $m$ for Nakagami-$m$ fading, and the same $k$ for Weibull fading), RVs $X_n$ are independent and identically distributed (i.i.d.) and their pdf $f_X(x)$ and cdf $F_X(x)$ are given by $f_{h_n}(x)$ and $F_{h_n}(x)$ in Table \ref{tab:pdfs_cdfs_hn} after setting $\Omega_n=1$. 
In the following, we analyze the per-user ergodic rate and average harvested energy achieved with order-based N-SNR scheduling. The pdf of the $j^{\text{th}}$ order statistics of the i.i.d. RVs $X_n$ is a special case of (\ref{eq:pdf_order_statistics_ind}) and can be written as \cite[eq. 2.1.6]{Order_Statistics_David_Nagaraja}
\begin{equation}
f_{X_{(j)}}(x)=N\binom{N-1}{j-1}f_X(x)[F_X(x)]^{j-1}[1-F_X(x)]^{N-j}.
\label{eq:pdf_order_statistics}
\end{equation}
\begin{table}[!t]
\caption{High-SNR approximation for the ergodic per-user rate over i.n.d. Nakagami-$m$, Weibull, and Ricean fading channels using order-based N-SNR scheduling. The exact ergodic per-user rate is shown for Rayleigh fading. The set $\mathcal{I}_{m,l}$ is defined in (\ref{eq:I_m_l}).} 
\label{tab:avg_Capacity_NSNR}
\begin{tabular}{@{}ll@{}}\toprule 
Fading Model & $\E\left[C_{j,\U_n}\right]$ \\ \midrule \addlinespace[0.5em]
	 		Nakagami-$m$ & 		
		$\begin{aligned}&\frac{1}{\Gamma(m)\ln(2)}\binom{N-1}{j-1}\sum\limits_{l=N-j}^{N-1} (-1)^{l-N+j} \binom{j-1}{N-l-1} \,\frac{l!}{(1+l)^m}\\
&\sum\limits_{\mathcal{I}_{m,l}}\left(\prod\limits_{s=0}^{m-1}\frac{\left(\frac{1}{s!(1+l)^s}\right)^{i_s}}{i_s!}\right)\Gamma\left(m+\sum\limits_{s=0}^{m-1} s\,i_s\right)\left(\psi\left(m+\sum\limits_{s=0}^{m-1} s\,i_s\right)+\ln\left(\frac{\bar{\gamma}_n}{m(1+l)}\right)\right)\\
\end{aligned}$\\ \addlinespace[0.5em]
 Weibull
&$\frac{1}{\ln(2)}\binom{N-1}{j-1}\sum\limits_{l=0}^{j-1}\frac{(-1)^l\binom{j-1}{l}}{N-j+l+1} \Bigg(\ln(\bar{\gamma}_n)-\frac{1}{k}\bigg(\ln\Big((N-j+l+1)\Gamma(1+\frac{1}{k})^k\Big)+\text{C}\bigg)\Bigg)$\\ \addlinespace[0.5em]
 Ricean
&$\frac{1}{\ln(2)}\binom{N-1}{j-1}\sum\limits_{l=0}^{j-1}\frac{(-1)^l\binom{j-1}{l}}{N-j+l+1} \Bigg(\ln(\bar{\gamma}_n)-\frac{1}{\mu'}\bigg(\ln\Big((N-j+l+1)\beta\Big)+\text{C}\bigg)\Bigg)$
\\ \addlinespace[0.5em]
 Rayleigh &		$\frac{1}{\ln(2)}\binom{N-1}{j-1}\sum\limits_{l=0}^{j-1} \frac{(-1)^{l} \binom{j-1}{l}}{{(N-j+l+1) }} \e^{\frac{(N-j+l+1)}{\bar{\gamma}_n}} 
		\E_1\left(\frac{(N-j+l+1)}{\bar{\gamma}_n}\right)$ \\		\bottomrule
		  \end{tabular}	
\end{table}
\begin{table}[!htbp]
  \caption{Average per-user harvested energy for i.n.d. fading channels using order-based N-SNR scheduling. The set $\mathcal{I}_{m,l}$ is defined in (\ref{eq:I_m_l}).}
\label{tab:avg_EH_NSNR}
		\begin{tabular}{@{}ll@{}}\toprule    
		Fading Model & $\E[EH_{j,\U_n}]=\eta P \Omega_n \left(1-\frac{1}{N}\E[X_{(j)}]\right)$\\ \midrule \addlinespace[0.5em]
		Nakagami-$m$ &
		$\eta P\Omega_n\left(1- \frac{1}{\Gamma(m+1)}\binom{N-1}{j-1}\sum\limits_{l=N-j}^{N-1}(-1)^{l-N+j}\binom{j-1}{N-l-1}\frac{l!}{(1+l)^{m+1}}
\sum\limits_{\mathcal{I}_{m,l}}\left(\prod\limits_{s=0}^{m-1}\frac{\left(\frac{1}{s!(1+l)^s}\right)^{i_s}}{i_s!}\right)\Gamma\left(m+1+\sum\limits_{s=0}^{m-1}s\, i_s\right)\right)
$\\ \addlinespace[0.5em]
		Weibull&
$\eta P\Omega_n \left(1 - \binom{N-1}{j-1}\sum\limits_{l=0}^{j-1}{(-1)^l\binom{j-1}{l} (N-j+l+1)^{-\left(1+\frac{1}{k}\right)}}\right)
$\\ \addlinespace[0.5em]
		Ricean&
$\eta P\Omega_n \left(1 - \binom{N-1}{j-1}\sum\limits_{l=0}^{j-1}{(-1)^l\binom{j-1}{l} (N-j+l+1)^{-\left(1+\frac{1}{\mu'}\right)}}\right)$
\\ \addlinespace[0.5em]
		Rayleigh&
$\eta P\Omega_n\left(1-\frac{1}{N}\sum\limits_{l=N-j+1}^{N}{\frac{1}{l}}\right)$\\ \bottomrule
\end{tabular}		
\end{table}
\\Hence, the ergodic rate of user $n$ is obtained as
\begin{equation}
\E\left[C_{j,\U_n}\right]=\frac{1}{N}\int\limits_0^\infty{\log_2\left(1+\bar{\gamma}_nx\right)f_{X_{(j)}}(x) \dd x},
\label{eq:Ergodic_capacity_NSNR}
\end{equation}
where $\bar{\gamma}_n\definedas \bar{\gamma}\Omega_n$ is the average SNR of user $n$, and  $\frac{1}{N}$ is the probability that the normalized channel of user $n$ has the $j^{\text{th}}$ order, since the normalized channels are i.i.d.
The average harvested energy of user $n$ is 
\begin{equation}
\E\left[EH_{j,\U_n}\right]=\eta P \Omega_n\int\limits_{0}^{\infty}{x\left(f_{X}(x)- \frac{1}{N}f_{X_{(j)}}(x)\right) \dd x}=\eta P \Omega_n\left[1-\frac{\E[X_{(j)}]}{N}\right],
\label{eq:EH_Un_NSNR}
\end{equation}
where we used $f_X(x)=\frac{1}{N}\sum_{j=1}^N f_{X_{(j)}}(x)$ and the fact that a user harvests energy when its normalized channel power is not the $j^{\text{th}}$ ordered one. 

Closed-form expressions for the per-user average rate and harvested energy in (\ref{eq:Ergodic_capacity_NSNR}) and (\ref{eq:EH_Un_NSNR}) are derived in Appendix \ref{app:proof_EH_C_NSNR} for Nakagami-$m$ and Weibull fading. Results for Ricean fading are deduced from those for Weibull fading by replacing $k$ by $\mu'$ and $\left(\Gamma\left(1+\frac{1}{k}\right)\right)^{k}$ by $\beta\definedas \e^{\nu(\sqrt{2K})}\sqrt{2(K+1)}^{\mu(\sqrt{2K})}$, cf. Section \ref{ss:channel_model}. High-SNR approximations for the ergodic rates are provided in Table \ref{tab:avg_Capacity_NSNR} for Nakagami-$m$, Weibull, and Ricean fading. The exact ergodic rate is shown for Rayleigh fading. Table \ref{tab:avg_EH_NSNR} shows the per-user average harvested energies for the considered fading channels, which are valid for all SNR regimes.
\vspace{-0.2cm}
\subsection{Order-based ET scheduling}
\label{s:controllable_ET_scheduling}
From the users' perspective, fairness in the sense that all users get the same average throughput may be preferred over proportional fairness. In this section, we design a resource allocation algorithm which does not only allocate ET to all users but also enables trading the ET value for the average amount of energy harvested by the users.
\subsubsection{Order-based ET Scheduling Algorithm}
Unlike in the conventional ET scheduling scheme discussed in Section \ref{ss:ET} which schedules the user having the smallest moving-average throughput among the set of \emph{all} users, in the proposed order-based ET scheduling scheme, a user is eligible for selection only if the order of its N-SNR falls into a specific set of allowed orders $\Sa\subset \{1,\ldots,N\}$, with $|\Sa|>1$. Among these eligible users, the selected user is the one which has the smallest throughput so far. This potentially leads to a controllable ET for all users. For example, if $\Sa=\{1,\ldots, N\}$, then all users are eligible for being scheduled by the AP and the scheme reduces to conventional ET scheduling as described in Section \ref{ss:ET}. If, however,  $\Sa= \{1,\ldots,\lfloor\frac{N}{2}\rfloor\}$, then the scheduled user is always among that half of the users, which have the lowest instantaneous N-SNRs. Hence, with this set, the average amount of energy harvested by the users  is expected to increase at the expense of a reduction in the ET compared to  conventional ET scheduling. This is because a user from the set of low N-SNR users is scheduled for data reception and the users having relatively high N-SNRs are selected for EH. Furthermore, sets $\Sa=\{1,2\}$ and $\Sa=\{N-1,N\}$ provide the extreme cases for the R-E tradeoff, by providing the minimum and the maximum possible ETs, respectively. Considering all cases with $|\Sa|>1$, there are in total $\sum_{k=2}^{N}\binom{N}{k}=\sum_{k=0}^{N}\binom{N}{k}- (N+1)=2^N -(N+1)$ possible choices for the set $\Sa$. The system designer can choose the set $\Sa$ which results in a desirable R-E tradeoff for the respective application.
We note that depending on the choice of $\Sa$ and the average channel power of the users $\Omega_n$, ET scheduling may not always be feasible. This issue is investigated later in Theorem \ref{theo:feasibility_conditions} in detail. 

To describe the scheduling algorithm, we define $O_{\U_n}$ as the order of the N-SNR of user $n$, where $O_{\U_n}\in\{1,\ldots,N\}$. In time slot $t$, the order-based ET scheduler selects for information transmission the user $n^*$ that satisfies  
 \begin{equation}
n^*=\argmin\limits_{O_{\U_n}\in \Sa} r_n(t-1),
\label{eq:order_ET_controllable_rule}
\end{equation}
where $r_n(t-1)$ is the throughput of user $n$ averaged over all previous time slots up to slot $t-1$. The throughput of the users is updated recursively as in (\ref{eq:throughput_update}).

\subsubsection{Performance Analysis}
In the following, we derive the achievable ET as well as the per-user average harvested energy for the order-based ET scheduling scheme. The average rate of user $n$ can be formulated as
\begin{equation}
\E[C_{\U_n}]=\E[C_{\U_n}|O_{\U_n}\in \Sa] \times \text{Pr}(O_{\U_n}\in \Sa).
\end{equation}
Since $O_{\U_n}$ takes values in $\{1,\ldots,N\}$ with equal probability $\forall\,n$, all users visit the set $\Sa$ with the same probability given by $ \text{Pr}(O_{\U_n}\in \Sa)=\frac{|\Sa|}{N}$. Thus, the average rate of user $n$ can be expressed as
\begin{equation}
\begin{aligned}
\E[C_{\U_n}]&=\frac{\left|\Sa\right|}{N} \int\limits_{0}^\infty\log_2(1+\bar{\gamma}_n x) \left(\frac{1}{|\Sa|}\sum\limits_{j\in\Sa}f_{X_{(j)}}(x) \right) \dd x \times \text{Pr}(\U_n|O_{\U_n}\in\Sa)\\
&=\sum\limits_{j\in\Sa}\E[C_{j,\U_n}]\Big|_{\text{N-SNR}}\text{Pr}(\U_n|O_{\U_n}\in\Sa),
\end{aligned}
\end{equation}
where $\frac{1}{|\mathcal{S}_{\text{a}}|}f_{X_{(j)}}(x)$ is the likelihood function that the order of the normalized channel of user $n$ is $j$ given that $j\in\Sa$, and $\text{Pr}(\U_n|O_{\U_n}\in\Sa)=\text{Pr}(n^*=n|O_{\U_n}\in\Sa)$ is the probability that user $n$ is scheduled given that it is eligible for scheduling. $\E[C_{j,\U_n}]\Big|_{\text{N-SNR}}$ is the average rate achieved by user $n$ using the order-based N-SNR scheduling scheme, cf. (\ref{eq:Ergodic_capacity_NSNR}).

In order to write the average rate of user $n$ in terms of its unconditional probability of being scheduled, we define $p_n$ as the probability that user $n$ is scheduled and use
\begin{equation}
p_n\definedas\text{Pr}(\U_n)=\text{Pr}(\U_n|O_{\U_n}\in\Sa) \text{Pr}(O_{\U_n}\in\Sa)=\text{Pr}(\U_n|O_{\U_n}\in\Sa)\frac{|\Sa|}{N}.
\label{eq:pn_pn_conditioned}
\end{equation}
Hence, the average rate of user $n$ reduces to
\begin{equation}
\E[C_{\U_n}]=\frac{N}{|\Sa|}\sum\limits_{j\in\Sa}\E[C_{j,\U_n}]\Big|_{\text{N-SNR}}p_n\req r,\quad \forall n\in\{1,\ldots,N\},
\label{eq:Avg_Capacity_Un_ET_step}
\end{equation}
where the average rate of all users is forced to be equal to $r$. Thus, the probability of channel access for user $n$ can be expressed as $p_n=r\Big/\left(\frac{N}{|\Sa|}\sum\limits_{j\in\Sa}\E[C_{j,\U_n}]\Big|_{\text{N-SNR}}\right)$, and since $\sum\limits_{n=1}^{N}{p_n}=1$ must hold, the resulting equal throughput $r$ reduces to
\begin{equation}
r=\frac{1}{\frac{1}{N}\sum\limits_{n=1}^{N}{\frac{1}{\frac{1}{|\Sa|}\sum\limits_{j\in\Sa}\E[C_{j,\U_n}]\Big|_{\text{N-SNR}}}}}.
\label{eq:controllable_ET}
\end{equation}
That is, for the order-based ET scheme, the equal throughput achieved by all users is given by the harmonic mean of the arithmetic means of the users' N-SNR rates $\E[C_{j,\U_n}]\Big|_{\text{N-SNR}}$ over all $j\in\Sa$. This indicates that the higher the orders in $\Sa$, the larger the resulting arithmetic means and consequently the larger the resulting ET $r$. Moreover, the harmonic mean indicates that the user having the worst average channel will have a dominant effect on the resulting ET. Using the closed-form expressions for the average rate of the order-based N-SNR scheme $\E[C_{j,\U_n}]\Big|_{\text{N-SNR}}$ for the considered fading channels in Table \ref{tab:avg_Capacity_NSNR}, the achievable ET in (\ref{eq:controllable_ET}) can be obtained in closed-form. The probability of channel access required for user $n$ to achieve the same throughput as the other active users is obtained from (\ref{eq:Avg_Capacity_Un_ET_step}) and (\ref{eq:controllable_ET}) as 
\begin{equation}
p_n=\left(\sum\limits_{i=1}^{N}{\frac{\sum\limits_{j\in\Sa}\E[C_{j,\U_n}]\Big|_{\text{N-SNR}}}{\sum\limits_{j\in\Sa}\E[C_{j,\U_i}]\Big|_{\text{N-SNR}}}}\right)^{-1}, \quad \forall n\in\{1,\ldots,N\}.
\label{eq:pn_controllable_ET}
\end{equation} 
As mentioned before, for certain combinations of $\Sa$ and $\Omega_n,\,\, n = 1,\ldots,N$, the order-based ET scheduling algorithm may fail to provide all users with the same throughput. In particular, the set of
scheduling probabilities $p_n,\,\,\,n = 1,\ldots,N$, in (\ref{eq:pn_controllable_ET}) required for the users to achieve ET may be infeasible. In the following theorem, we provide necessary and sufficient conditions for ET-feasibility of the proposed order-based ET scheduling algorithm.
\begin{theorem}
The order-based ET scheme with $\left|\Sa\right|>1$ is ET-feasible if and only if
\begin{equation}
\begin{aligned}
p_n &\leq \frac{|\Sa|}{N}, &&\quad \forall n\in\{1,\ldots,N\},\\
\sum\limits_{l=1}^{L} p_{n_{l}} &\leq \frac{\binom{N-1}{|\Sa|-1}L+\binom{L}{|\Sa|}(1-|\Sa|)}{\binom{N}{|\Sa|}},  &&\quad\begin{aligned}&\forall (n_1,\ldots,n_{L})\in\mathcal{C}_{L},\\
&\forall L=|\Sa|,\ldots,N,\end{aligned}
\end{aligned}
\label{eq:feasibility_conditions}
\end{equation}
where $p_n$ is given by (\ref{eq:pn_controllable_ET}) and $\mathcal{C}_{L}$ is the set of all $\binom{N}{L}$ combinations $(n_1,\ldots,n_{L})$ of $\{1,\ldots,N\}$.
\label{theo:feasibility_conditions}
\end{theorem}
\begin{proof} Please refer to Appendix \ref{app:feasibility_conditions_proof}. \end{proof}
\begin{remark}
Note that the second feasibility condition is always satisfied for $L\!=\!N$ as it reduces to $\sum_{n=1}^{N} p_n\leq 1$ which is satisfied with equality by definition. Furthermore, the conventional ET scheme, where $|\Sa|=\!N$, is always ET-feasible, since the first condition reduces to $p_n\leq 1$ which is always satisfied from (\ref{eq:pn_controllable_ET}) and the second condition reduces to the case $L=N$ and hence is always satisfied.   \end{remark}
\begin{remark}
In most practical scenarios, the order-based ET scheduling algorithm is ET-feasible. ET-infeasibility occurs when the mean channel power gains $\Omega_n$ of the users differ by many orders of magnitude. For example, a scenario with 4 users having Rayleigh fading channels with $\Omega_n=$$1,\,1,\,10^{-10},$ and $10^{-10}$ and with $\Sa=\{3,4\}$ is ET-feasible since the required scheduling  probability set $p_n=\{0.0884,0.0884,0.4116,0.4116\}$ satisfies the conditions in Theorem \ref{theo:feasibility_conditions}. In contrast, the same scenario but with $\Omega_n=1,\,1,\,10^{-11},$ and $10^{-11}$ is ET-infeasible since the required scheduling probability set $p_n=\{0.0603,0.0603,0.4397,0.4397\}$ does not satisfy the second feasibility condition in Theorem \ref{theo:feasibility_conditions} for $L=|\Sa|=2$. In this case, it can be verified by simulations that ET is not achieved as the ergodic rates of the users achieved with the algorithm in (\ref{eq:order_ET_controllable_rule}) are $3.5285,\,3.5285,\,2.415,$~and~$2.415$ \unit[]{bits/(channel use)}, respectively.
\label{remark:Mostly_feasible}
\end{remark}

Next, we analyze the average amount of harvested energy per user. Define $\Sac$ as the complement of the set $\Sa$ with respect to the set $\{1,\ldots,N\}$, then we can write the average harvested energy of user $n$ as
\begin{equation}
 \E\left[EH_{\U_n}\right]=\E\left[EH_{\U_n}|O_{\U_n}\in \Sac\right] \times \text{Pr}(O_{\U_n}\in \Sac)+\E\left[EH_{\U_n}|O_{\U_n}\in \Sa\right] \times \text{Pr}(O_{\U_n}\in \Sa).
\end{equation}
Since users whose N-SNR orders are in $\Sac$ will harvest energy with probability one, whereas users whose N-SNR orders are in $\Sa$ will only harvest if they are not scheduled to receive information, then we have $\E\left[EH_{\U_n}\right]$ 
\begin{align}
&=\eta P \Omega_n \left[\int\limits_0^\infty x \frac{1}{|\Sac|}\sum\limits_{j\in\Sac}f_{X_{(j)}}(x) \dd x \times \frac{|\Sac|}{N}+
\int\limits_0^\infty x \frac{1}{|\Sa|}\sum\limits_{j\in\Sa}f_{X_{(j)}}(x) \left(1-\frac{p_n N}{|\Sa|}\right) \dd x \times \frac{|\Sa|}{N}\right]\notag\\
&=\eta P \Omega_n \left[ \frac{1}{N}\sum\limits_{j=1}^N \E\left[X_{(j)}\right] - \frac{p_n}{|\Sa|}\sum\limits_{j\in\Sa}\E\left[X_{(j)}\right] \right]=\eta P \Omega_n \left[ 1-\frac{p_n}{|\Sa|}\sum\limits_{j\in\Sa}\E\left[X_{(j)}\right] \right],
\label{eq:avg_EH_controllable_ET}
\end{align}
where from (\ref{eq:pn_pn_conditioned}), $\frac{p_n N}{|\Sa|}$ is the conditional probability that user $n$ is scheduled given that $O_{\U_n}\in \Sa$. In the last step, the unity term is obtained using $\sum_{j=1}^{N} X_{(j)}\!=\!\sum_{n=1}^{N} X_n$, thus $\sum_{j=1}^{N} \E\left[X_{(j)}\right]\!=\!\sum_{n=1}^{N} \E\left[X_n\right]\!=\!N$, since the normalized channel powers $X_n$ are unit-mean RVs $\forall \,n\!=\!1,\!\ldots\!,N$. The mean $\E\left[X_{(j)}\right]$ was already obtained as part of the closed-form expression of the per-user average harvested energy for order-based N-SNR scheduling for the considered fading channels, cf. Table \ref{tab:avg_EH_NSNR}. Thus, the average per-user harvested energy for the order-based ET scheme in (\ref{eq:avg_EH_controllable_ET}) can be written in closed-form as well. Note that for all users to achieve the same throughput, users with bad channel conditions are selected more often for information reception (have higher $p_n$) than those who have better channel conditions. Hence, bad-channel users have less chances for energy harvesting. Furthermore, when a bad-channel user (with small $\Omega_n$) harvests, it harvests less energy than a good-channel user. In conclusion, the order-based ET scheduling scheme provides proportional fairness in terms of the harvested energy, which may also be observed from (\ref{eq:avg_EH_controllable_ET}).

\section{Numerical and Simulation Results}
\label{s:simulation_results}
The proposed scheduling schemes have been simulated for an indoor environment operating in the ISM band at a center frequency of $\unit[915]{MHz}$ (wavelength of 0.328\,m) and a bandwidth of $\unit[26]{MHz}$. The resulting  noise power is $\sigma^2=\unit[-96]{dBm}$ at the receivers of all users. We adopt the indoor path loss model in \cite{PathLoss_indoor_Rappaport_1992} for the case when the AP and the UTs are on the same floor (i.e., a path loss exponent of 2.76 is used, cf. \cite[Table I]{PathLoss_indoor_Rappaport_1992}). We assume an AP transmit power of $P=\unit[1]{W}$, an antenna gain of $\unit[10]{dBi}$ at the AP and $\unit[2]{dBi}$ at the UTs, and an RF-to-DC conversion efficiency of $\eta=0.5$. First, we consider a system with $N=7$ users having mean channel power gains $\Omega_n=n\times10^{-5}$, $n=1,\ldots,7$, which corresponds to an AP-UT  distance range of $\unit[2.27]{m}$ to $\unit[4.6]{m}$. 
\begin{figure}[!t]
\centering
\includegraphics[width=0.62\textwidth]{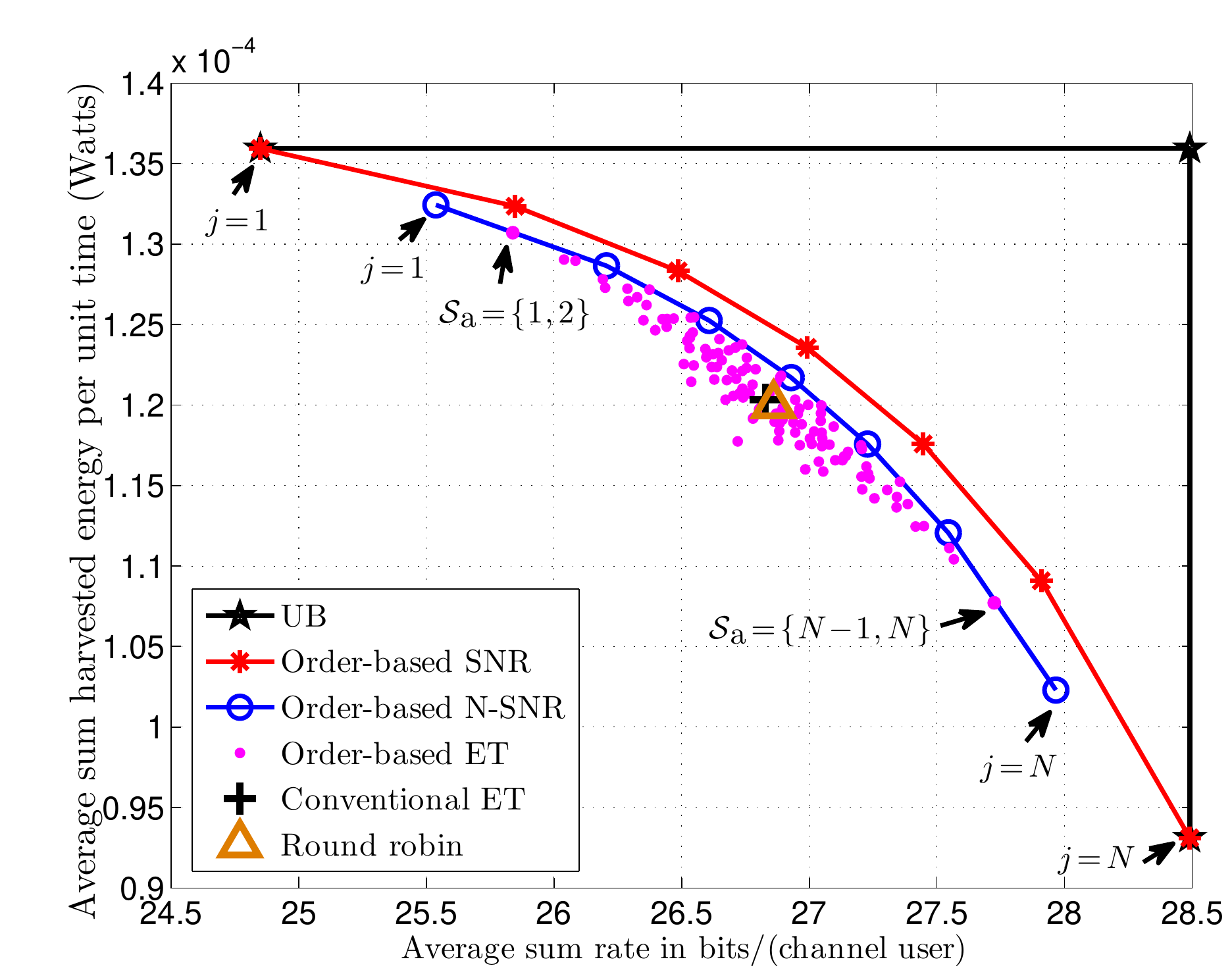}
\caption{Average sum rate and total harvested energy for different scheduling schemes for $N=7$ users over i.n.d. Nakagami-$m$ fading  with $m=3$.}
\label{fig:publication_curve_Nakagami_sum_rate_sum_energy_all_schemes}
\end{figure}

Fig. \ref{fig:publication_curve_Nakagami_sum_rate_sum_energy_all_schemes} shows the sum rate $\sum_{n=1}^N \E\left[C_{j,\U_n}\right]$ vs. the total harvested energy  $\sum_{n=1}^N \E[EH_{j,\U_n}]$ tradeoff behavior of the proposed order-based schemes over i.n.d. Nakagami-$m$ fading channels with $m=3$. Only closed-form results obtained from Tables \ref{tab:avg_Capacity_SNR}-\ref{tab:avg_EH_NSNR} together with (\ref{eq:controllable_ET}) and (\ref{eq:avg_EH_controllable_ET}) are shown since they perfectly match the simulated results. Fig. \ref{fig:publication_curve_Nakagami_sum_rate_sum_energy_all_schemes} shows that the smaller the values of $j$ and the lower the orders in the set $\Sa$, the higher the amount of total harvested energy at the expense of a reduced sum rate. This is because the good states of the channel are utilized for EH rather than for ID. Owing to the fact that the order-based SNR scheme provides the maximum possible sum rate for $j=N$ and the maximum possible total harvested energy for $j=1$, connecting the horizontal line passing through the R-E point for $j=1$ and the vertical line passing through the R-E point for $j=N$ of this scheme provides a valid rectangular upper bound for the R-E region of any scheduling scheme (denoted by UB). We also show the R-E points of two baseline schemes, namely, RR (cf. Section \ref{ss:RR}) and conventional ET (cf. Section \ref{ss:ET}), which are not biased towards information transfer nor towards energy transfer.

The proposed order-based schemes provide a R-E tradeoff in discrete steps. That is, for the order-based SNR/N-SNR schemes, there are $N$ tradeoff points corresponding to selection orders $j=1,\ldots,N$. For the order-based ET scheme, there are a total of $2^N-(N+1)=120$ possible sets $\Sa$ with $|\Sa|>1$, cf. Section \ref{s:controllable_ET_scheduling}, all of which are ET-feasible for the considered setup (as expected, see Remark \ref{remark:Mostly_feasible}). In Fig. \ref{fig:publication_curve_Nakagami_sum_rate_sum_energy_all_schemes}, we plot the achievable R-E points of the order-based ET scheduler for all possible sets to show that some ``good" sets provide a better R-E tradeoff than other sets (i.e., they provide higher harvested energy for the same sum rate). With the knowledge of the number of users $N$ and the channel statistics ($\Omega_n$ and $m$), the AP may obtain the ``good" $\Sa$ sets offline with the aid of the closed-form expressions in (\ref{eq:controllable_ET}) and (\ref{eq:avg_EH_controllable_ET}). Note that the sets of two consecutive orders $\{i,i+1\}$ provide R-E tradeoff points in between the R-E tradeoff points of the order-based N-SNR scheme for orders $i$ and $i+1$.

It is observed that, as far as the sum rate and the total harvested energy are concerned, the R-E tradeoff of order-based SNR scheduling is better than that of the order-based N-SNR scheduling which is better than that of the ET scheduling. This is due to the proportional fairness and the ET constraints of the order-based N-SNR and the order-based ET schedulers, respectively. If we consider the extreme tradeoff points of the order-based schemes, we find that going from $j=N$ to $j=1$ for the order-based SNR and N-SNR schemes, and from $\Sa=\{N-1,N\}$ to $\Sa=\{1,2\}$ for the order-based ET scheme, the total average harvested energy increases by $45.98\%$, $29.45\%$, and $21.35\%$, respectively, at the expense of a reduction in the ergodic sum rate of only $12.78\%$, $8.68\%$, and $6.78\%$, respectively.
\begin{figure}[!t]
\centering
 \subfloat[Per-user ergodic rate]{\label{subfig:Publication_curve_Nakagami_ordered_SNR_per_user_Capacity}\includegraphics[height=0.26\textheight]{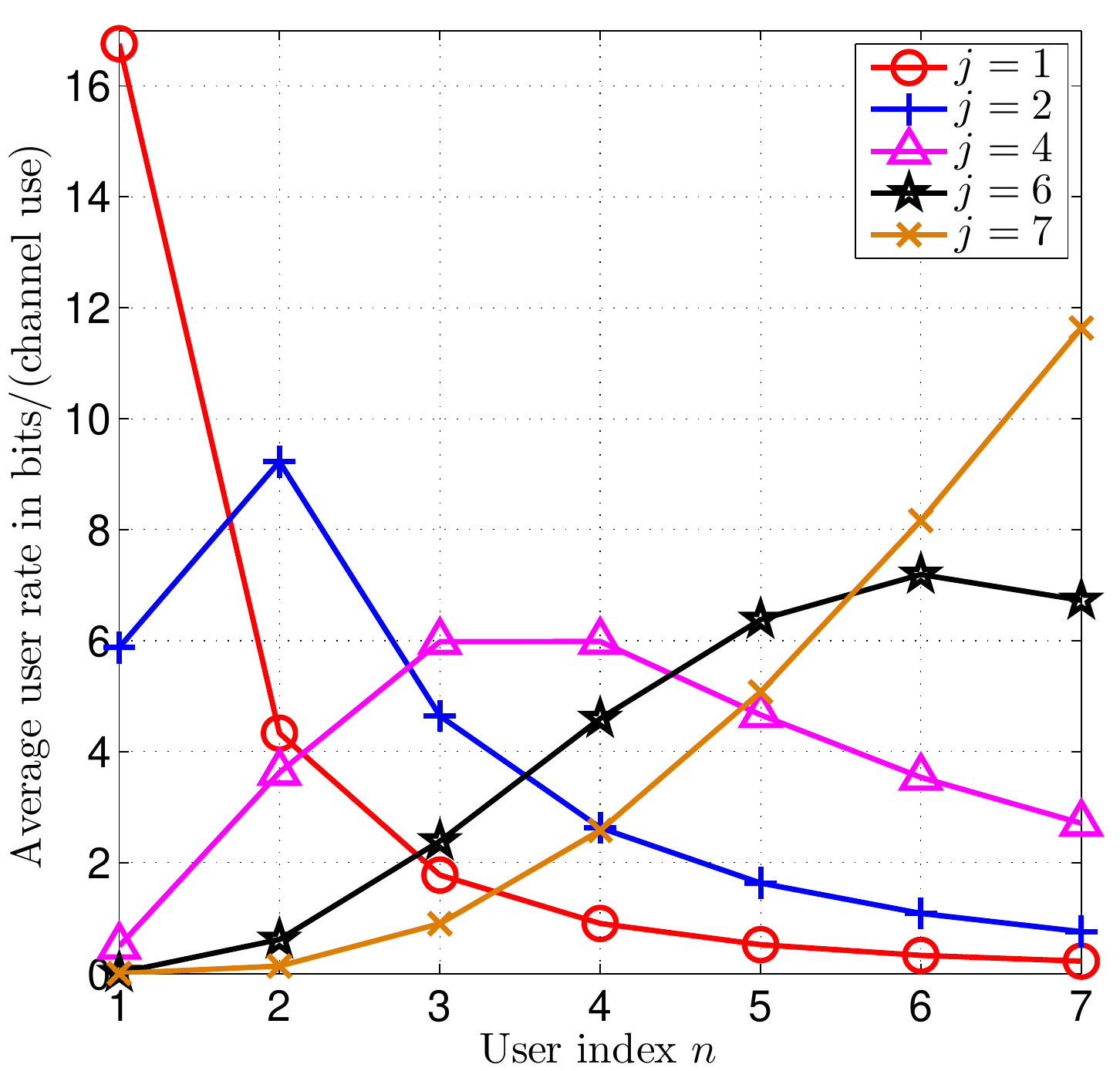}}\hspace{1cm}
 \subfloat[Per-user average harvested energy]{\label{subfig:Publication_curve_Nakagami_ordered_SNR_per_user_Energy}\includegraphics[height=0.27\textheight]{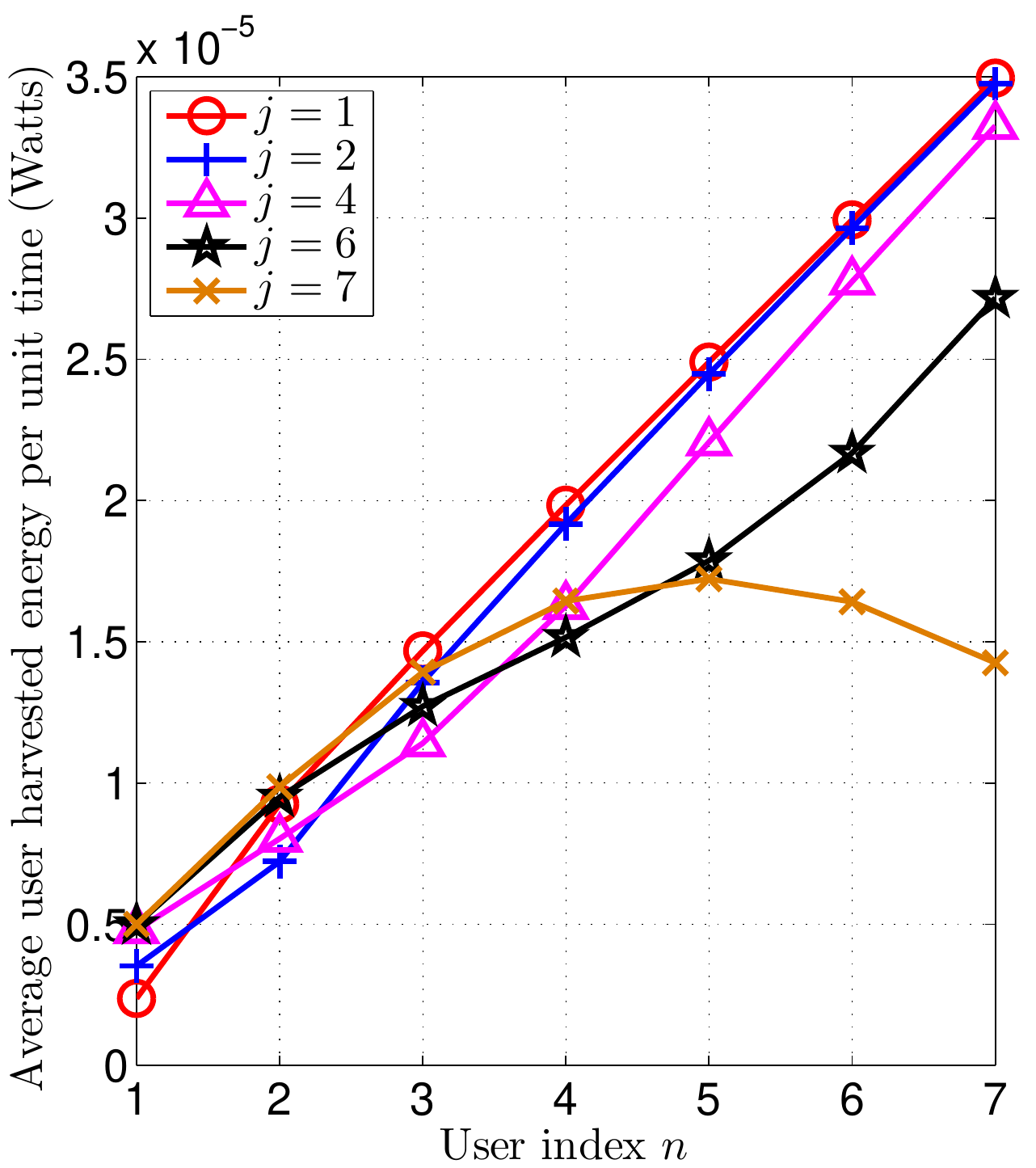}}
\caption{Average per-user rate and harvested energy for order-based SNR scheduling for $N=7$ users over i.n.d. Nakagami-$m$ fading  with $m=3$.}
\label{fig:Publication_curve_Nakagami_ordered_SNR_per_user_Capacity_energy}
\end{figure}

In Fig. \ref{fig:Publication_curve_Nakagami_ordered_SNR_per_user_Capacity_energy}, we show the \emph{per-user} average rate and harvested energy performance for order-based SNR scheduling over i.n.d. Nakagami-$m$ fading channels with $m=3$.  Fig. \ref{fig:Publication_curve_Nakagami_ordered_SNR_per_user_Capacity_energy} reveals that for any order $j$, the scheme does not provide fairness among the users, neither in terms of their data rates nor their harvested energies. The average amount of harvested energy and the average rate of a user depend on its average channel power gain, the selection order $j$ used, and the order of the user's average channel power gain relative to that of the other users. For the assumed values of $\Omega_n=n\times10^{-5}$, $n=1,\ldots,7$, and a selection order of $j$, the user with the $j^{\text{th}}$-ordered $\Omega_n$ achieves the highest rate. Similarly, the amount of energy harvested by a user depends on how often the user is selected and how much energy it can harvest when it is not selected. Hence, the main advantage of the order-based SNR scheme is not in its per-user performance, but rather its sum rate vs. total harvested energy performance which was shown to provide a better R-E tradeoff compared to the proposed fair schemes.

Fig. \ref{fig:Publication_curve_NSNR_ET_Rician} shows the performance of the order-based N-SNR and the order-based ET scheduling schemes over i.n.d. Ricean fading channels with Ricean factor $K=6$. The closed-form analytical results were obtained using Tables \ref{tab:avg_Capacity_NSNR} and \ref{tab:avg_EH_NSNR} together with (\ref{eq:controllable_ET}), (\ref{eq:pn_controllable_ET}) and (\ref{eq:avg_EH_controllable_ET}). It is observed that there is a minor difference between the simulated performance curves and the closed-form expressions which use the exponential approximation of the Marcum-Q function provided in Section \ref{ss:channel_model}.\footnote{The accuracy of the approximation was tested for up to $K=18$, where it was found that the higher the Ricean factor $K$ gets, the looser the approximation is. For example, when $K=18$, the closed-form expressions for the average per-user rate and harvested energy for the order-based N-SNR scheme differ from the simulated ones by at most $0.5\%$ and $0.65\%$, respectively.} Every point in each curve of Figs. \ref{subfig:Publication_curve_NSNR_Rician} and \ref{subfig:Publication_Curve_ET_Ricean} corresponds to the ergodic rate and the average harvested energy of a specific user. Since both order-based N-SNR and ET schemes provide proportional fairness among the users in terms of the average harvested energy, points corresponding to more harvested energy belong to a stronger-channel user. In Fig. \ref{fig:Publication_curve_NSNR_ET_Rician}, we also highlight the R-E curves of the worst- and the best-channel users.  

\begin{figure}[!t]
 \subfloat[Order-based N-SNR scheduling]{\label{subfig:Publication_curve_NSNR_Rician}\includegraphics[width=0.495\textwidth, trim= 0 0 1.2cm 0,clip]{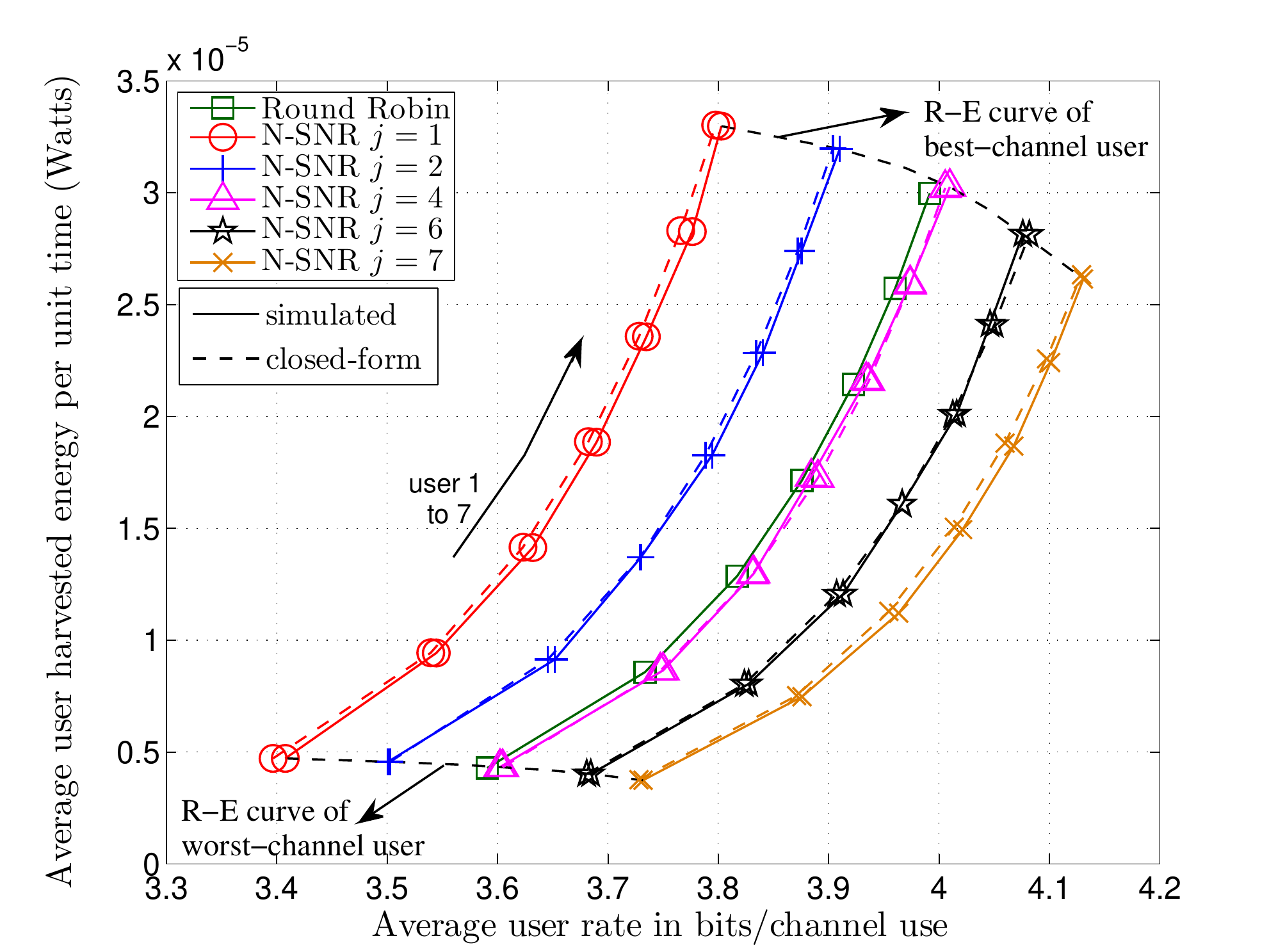}}
 \subfloat[Order-based ET scheduling]{\label{subfig:Publication_Curve_ET_Ricean}\includegraphics[width=0.495\textwidth]{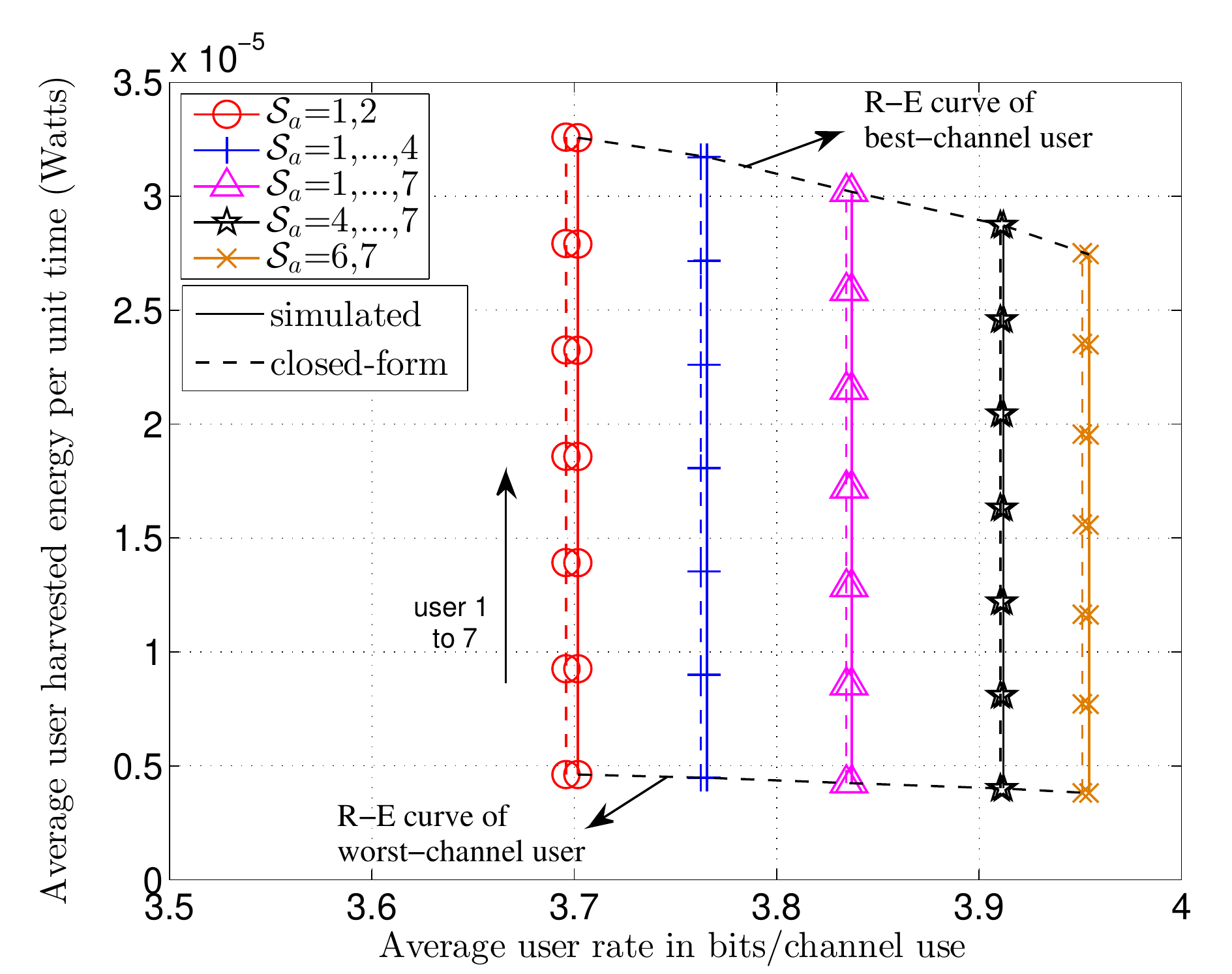}}
\caption{Rate-energy performance of the order-based N-SNR, the RR and the order-based ET scheduling schemes for $N=7$ users over i.n.d. Ricean fading with $K=6$.}
\label{fig:Publication_curve_NSNR_ET_Rician}
\end{figure}
In Fig. \ref{subfig:Publication_curve_NSNR_Rician}, it is observed that both the order-based N-SNR and the RR scheduling schemes achieve proportional fairness in terms of both the ergodic rate and the average amount of harvested energy, since all users are on average scheduled for the same number of time slots. Furthermore, the RR scheme is shown to perform in-between the order-based N-SNR curves as it is neither biased towards energy transfer nor towards information transfer. Moreover, for the order-based N-SNR scheme, by reducing $j$ in integer-steps from $N$ to $1$, we allow the users to harvest more energy at the expense of reducing their ergodic rates. For example, for the best-channel user, reducing $j$ from $N$ to $1$ leads to a $7.94\%$ reduction in rate and a $26.1\%$ increase in harvested energy. 

In Fig. \ref{subfig:Publication_Curve_ET_Ricean}, the order-based ET scheduling scheme is shown to provide all users with ET, and thus is ET-feasible for all considered sets $\Sa$ as can be verified using Theorem \ref{theo:feasibility_conditions}. It is observed that for the same $|\Sa|$, the lower the allowed orders in $\Sa$ get, the higher the achievable average harvested energy for all users at the expense of a reduced ET. By observing the performance of sets $\Sa=\{1,2\}$ and $\Sa=\{6,7\}$ which provide the extreme cases for the R-E tradeoff, we find that going from $\Sa = \{6, 7\}$ to $\Sa = \{1, 2\}$ leads to an increase of $18.6\%$ and $21\%$ in the amount of harvested
energy for the best- and the worst-channel users, respectively, at the expense of only a $6.33\%$ reduction in the ET. The performance curve for $\Sa=\{1,\ldots,N\}$ corresponds to the conventional ET scheduling scheme described in Section \ref{ss:ET}.
\begin{figure}[!t]
\centering
 \subfloat[Average sum rate]{\label{subfig:publication_curve_Weibull_ordered_NSNR_different_num_users_Capacity}\includegraphics[height=0.25\textheight]{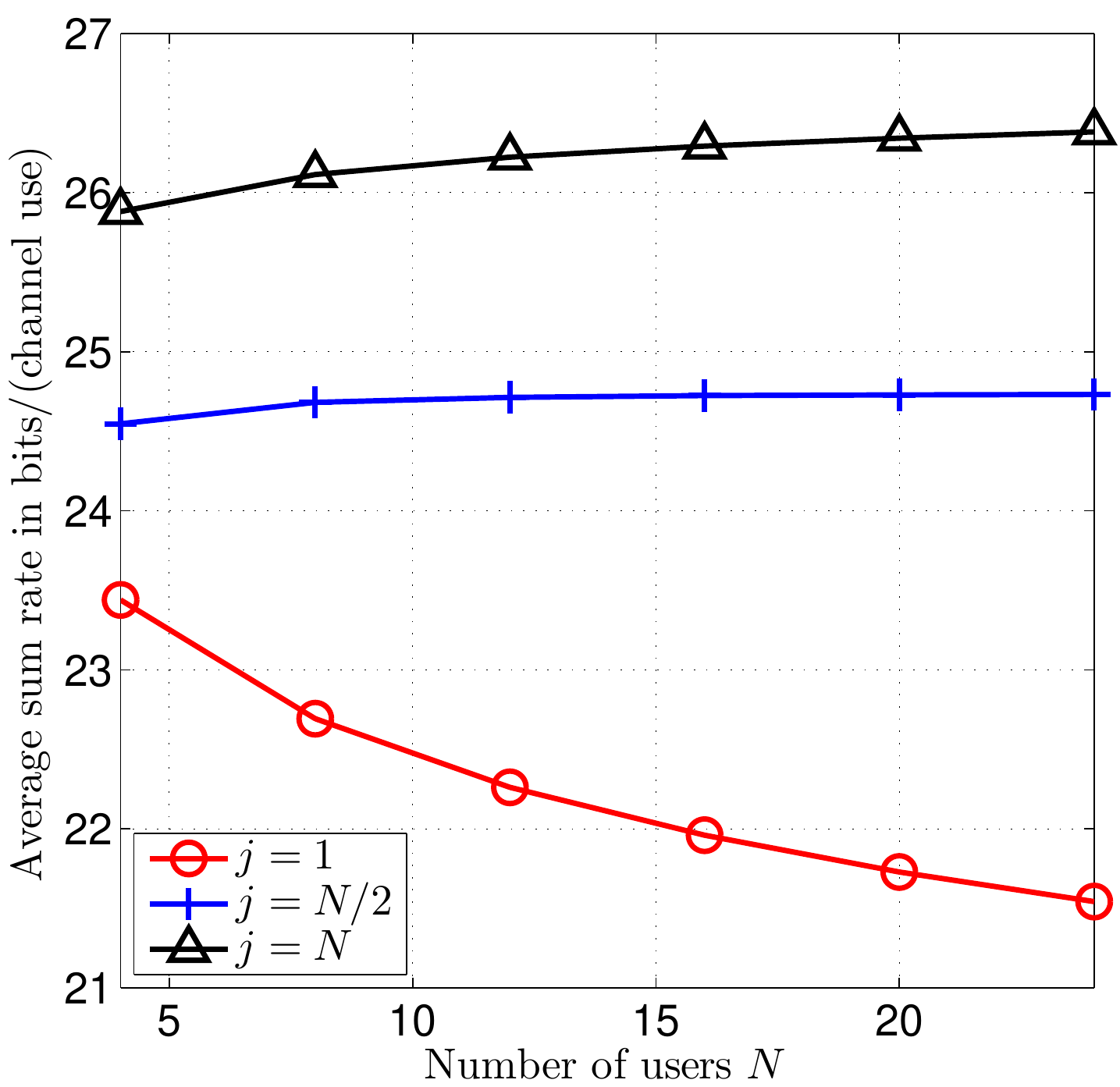}}\hspace{1cm}
 \subfloat[Average total harvested energy]{\label{subfig:publication_curve_Weibull_ordered_NSNR_different_num_users_Energy}\includegraphics[height=0.26\textheight]{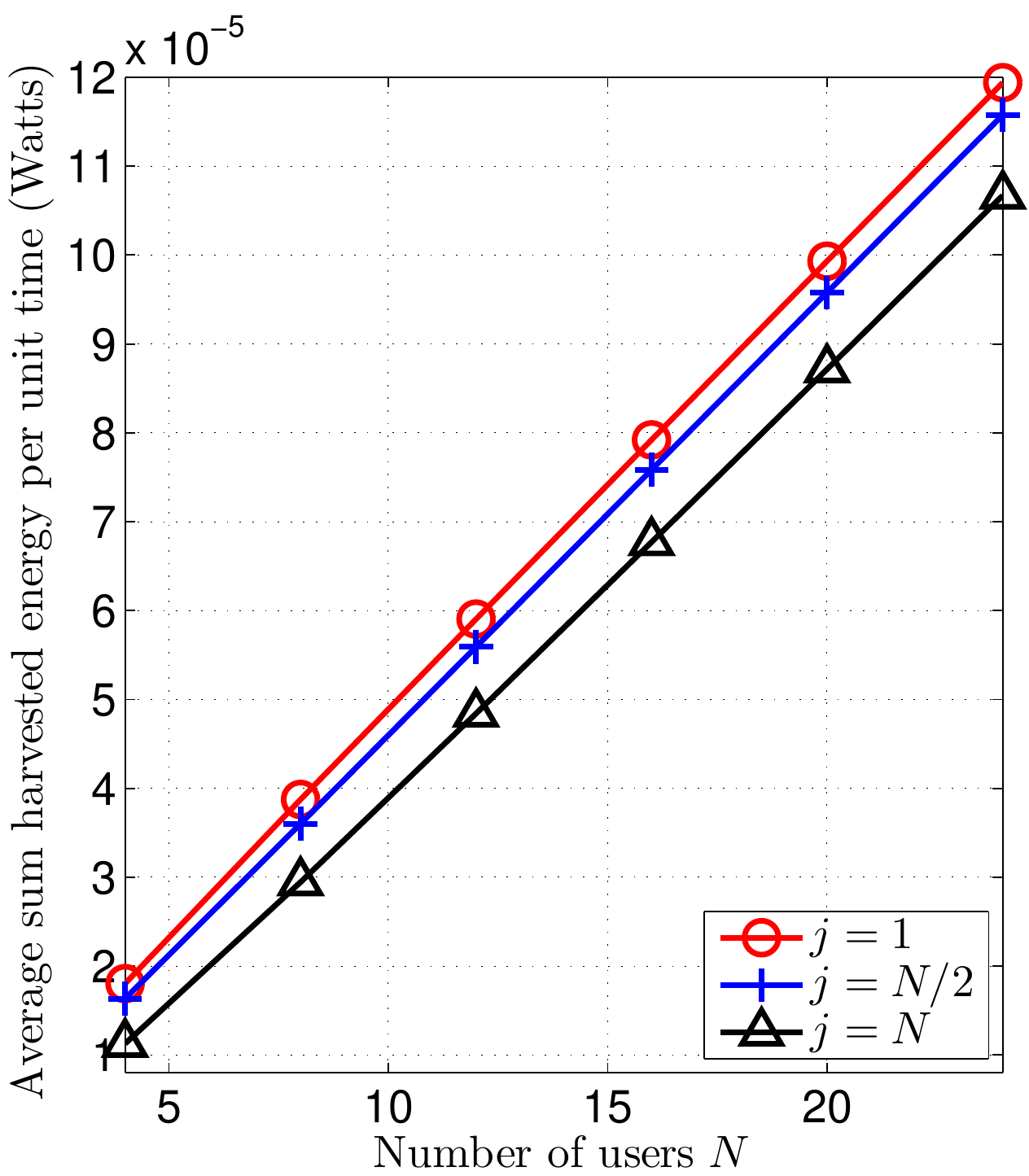}}
\caption{Average sum rate and total harvested energy for N-SNR scheduling for different numbers of users with $\Omega_n=\frac{n}{\frac{1}{N}\sum_{i=1}^N{i}}\times 10^{-5}$ over i.n.d. Weibull fading with $k=1.5$.}
\label{fig:publication_curve_ordered_NSNR_different_num_users_Weibull_sim_only}
\end{figure}
\begin{remark}
The selection order $j$ for the order-based SNR/N-SNR scheme and the set of orders $\Sa$ for the order-based ET scheme can be chosen to satisfy the R-E tradeoff desired by the users, within the limits of the feasible  R-E tradeoffs of the respective schemes. In particular, each user may send a feedback signal  to the AP which indicates its rate and energy needs. If, for example, a user is in high demand of energy, its feedback signal indicates that it needs to be excluded from the selection process to allow it to harvest as much energy as possible. However, it may not be feasible to satisfy the needs of all users, since they may have conflicting requirements. In this case, using the performance results in Tables \ref{tab:avg_Capacity_SNR}-\ref{tab:avg_EH_NSNR}, (\ref{eq:controllable_ET}) and (\ref{eq:avg_EH_controllable_ET}), the AP may choose the order $j$ or the set $\Sa$ which would satisfy the majority of the users.
\end{remark}

Fig. \ref{fig:publication_curve_ordered_NSNR_different_num_users_Weibull_sim_only} shows the effect of the number of users $N$ on the ergodic sum rate and the total harvested energy of the order-based N-SNR scheduling scheme for $j=1,\frac{N}{2},$ and $N$. Simulations were performed for a Weibull fading channel with shape parameter $k=1.5$. Only closed-form analytical results are shown as they perfectly match the simulated ones. For a fair comparison, the average channel power gain of the users is normalized to $10^{-5}$ for all $N$, i.e., $\Omega_n=\frac{n}{\frac{1}{N}\sum_{i=1}^N{i}}\times 10^{-5}$. We observe that for any order $j$,  the total average harvested energy increases with the number of users, since  having more users implies capturing more of the ambient RF energy. However, the effect of the number of users on the sum rate depends on the order $j$. To understand this effect, consider the extreme orders $j\!=\!1$ and $j\!=\!N$. For $j\!=\!N$, more users implies a higher probability for larger maximum N-SNR values, and therefore a higher average rate (MUD gain). However, for $j\!=\!1$, more users implies a higher probability for smaller minimum N-SNR values, and therefore a lower average rate results (MUD loss). For middle orders (e.g., $j=\frac{N}{2}$), there is hardly any gain or loss in the ergodic sum rate for $N\geq8$.
\begin{remark}
We note that for lower data rates, the clock frequency and the supply voltage of the UT circuits can be scaled down (a technique known as dynamic voltage scaling). This leads to a cubic reduction in power consumption because dynamic power dissipation depends on the square of the supply voltage and linearly on the frequency ($P_c\propto V^2 f$) \cite[p. 88]{wang2009electronic}. Hence, when the users can tolerate low data rates, the selection order $j$ of the order-based SNR/N-SNR scheduling or the orders in $\Sa$ for the order-based ET scheduling can be chosen small to allow the users to harvest more RF energy and simultaneously reduce their power consumption.
\end{remark}

\section{Conclusion}
\label{s:conclusion}
In this paper, new online multi-user scheduling schemes that enable the control of the tradeoff between the sum rate and the average amount of harvested energy are proposed for SWIPT systems. The proposed order-based N-SNR/ET scheduling schemes additionally ensure long-term proportional/equal-throughput fairness (for ET-feasible scenarios) in terms of the users' rates. Furthermore, both the order-based N-SNR and the order-based ET scheduling schemes ensure proportional fairness in terms of the amount of energy harvested by the users. We applied order statistics theory to analyze the per-user ergodic achievable rate and average amount of harvested energy of the proposed schemes and provided them  in closed-form for i.n.d. Nakagami-$m$, Ricean, Weibull, and Rayleigh fading channels. Furthermore, feasibility conditions required for the order-based ET scheme to achieve ET for all users were derived. Our results reveal that the lower the selection order for the order-based SNR/N-SNR scheme, or the lower the orders in $\Sa$ for the order-based ET scheme, the higher the average total amount of harvested energy at the expense of a reduced ergodic sum rate.  

\appendices
\vspace{-0.47cm}
\section{Analysis of order-based SNR scheduling}
\label{app:proof_EH_C_SNR} 
\subsection{Nakagami-$m$ Fading}
\label{app:EH_C_Nakagami_proof_SNR} 
Using the cdf of the channel power gain for Nakagami-$m$ fading in Table \ref{tab:pdfs_cdfs_hn}, the products $\prod\limits_{l=1}^{j-1}{F_{h_{i_l}}(x)}$ and $\prod\limits_{l=j}^{N-1}{\left(1-F_{h_{i_l}}(x)\right)}$, which appear in the per-user ergodic rate and average harvested energy in (\ref{eq:Order_SNR_per_user_Ergodic_capacity}) and (\ref{eq:Order_SNR_per_user_EH}), can be written as sums of products given by
\begin{equation}
\prod\limits_{l=1}^{j-1}{F_{h_{i_l}}(x)}=\sum\limits_{r=0}^{j-1}(-1)^r\sum\limits_{(c_1,\ldots,c_r)\in\mathcal{C}_{r,j}}\,\,\sum\limits_{s_1,\ldots,s_{r}=0}^{m-1}\,\,\frac{\prod\limits_{t=1}^r \lambda_{c_t}^{s_t}}{\prod\limits_{t=1}^r s_t!}\,x^{\sum\limits_{t=1}^r s_t}\e^{-\!\!\sum\limits_{t=1}^r\lambda_{c_t}x}
\label{eq:prod_cdfs_Nakagami}
\end{equation}
and \vspace{-0.5cm}
\begin{equation}
\prod\limits_{l=j}^{N-1}{\left(1-F_{h_{i_l}}(x)\right)}=\sum\limits_{s_j,\ldots,s_{N-1}=0}^{m-1}\,\,\frac{\prod\limits_{l=j}^{N-1} \lambda_{i_l}^{s_l}}{\prod\limits_{l=j}^{N-1} s_l!}\,x^{\sum\limits_{l=j}^{N-1} s_l}\e^{-\!\!\sum\limits_{l=j}^{N-1}\lambda_{i_l}x},
\label{eq:prod_1_cdfs_Nakagami}
\end{equation} 
where $\mathcal{C}_{r,j}$ is the set of all $\binom{j-1}{r}$ combinations $(c_1,\ldots,c_r)$ of $(i_1,\dots,i_{j-1})$, with $(i_1,\dots,i_{N-1})$ being one permutation in $\mathcal{P}_n$. Using (\ref{eq:Order_SNR_per_user_Ergodic_capacity}), (\ref{eq:prod_cdfs_Nakagami}), and (\ref{eq:prod_1_cdfs_Nakagami}), the per-user ergodic rate reduces to
\begin{equation}
\begin{aligned}
\E\left[C_{j,\U_n}\right]=\frac{\lambda_n^m}{\ln(2)\Gamma(m)}\sum\limits_{\mathcal{P}_n}\sum\limits_{r=0}^{j-1}&(-1)^r\sum\limits_{\mathcal{C}_{r,j}}\sum\limits_{s_1,\ldots,s_{r}=0}^{m-1}\,\,\sum\limits_{s_j,\ldots,s_{N-1}=0}^{m-1}\frac{\prod\limits_{t=1}^r \lambda_{c_t}^{s_t}}{\prod\limits_{t=1}^r s_t!}\frac{\prod\limits_{l=j}^{N-1} \lambda_{i_l}^{s_l}}{\prod\limits_{l=j}^{N-1} s_l!}\\
&\int\limits_0^\infty{\ln\left(1+\bar{\gamma}x\right)x^{\left(\sum\limits_{t=1}^r s_t+m-1+\sum\limits_{l=j}^{N-1} s_l\right)}\e^{-\left(\sum\limits_{t=1}^r\lambda_{c_t}+\lambda_n+\sum\limits_{l=j}^{N-1}\lambda_{i_l}\right)x}\dd x}.
\end{aligned}
\label{eq:Order_SNR_per_user_Ergodic_capacity_step1}
\end{equation}  
To rewrite (\ref{eq:Order_SNR_per_user_Ergodic_capacity_step1}) in a more compact form, we define the set $\mathcal{U}_{n,r}$ as
\begin{equation}
\begin{aligned}
\mathcal{U}_{n,r}=\{&(u_1,\ldots,u_{N-j+r})\mid(u_1,\ldots,u_r)=(c_1,\ldots,c_r) \textrm{ and } (u_{r+1},\ldots,u_{N-j+r})=(i_j,\ldots,i_{N-1}),\\ & \forall \textrm{ permutations } (i_1,\ldots,i_{N-1})\in \mathcal{P}_n \textrm{ and $\forall$ combinations } (c_1,\ldots,c_r) \textrm{ of } (i_1,\ldots,i_{j-1})\}.
\end{aligned}
\label{eq:set_Unr}
\end{equation} 
Thus, $\left|\mathcal{U}_{n,r}\right|=\left| \mathcal{P}_n\right|\left|\mathcal{C}_{r,j}\right|=\binom{N-1}{j-1}\binom{j-1}{r}$. In addition, we define the set $\mathcal{S}_{m,r}$ as
\begin{equation}
\mathcal{S}_{m,r}=\left\{(s_1,\ldots,s_{N-j+r})\mid s_t\in\{0,\ldots,m-1\},\,\forall t=\{1,\ldots,N-j+r\}\right\},
\label{eq:set_Sm}
\end{equation}
where $\left|\mathcal{S}_{m,r}\right|=m^{N-j+r}$. Hence, (\ref{eq:Order_SNR_per_user_Ergodic_capacity_step1}) can be written as $\E\left[C_{j,\U_n}\right]=$
\begin{equation}
\frac{\lambda_n^m}{\ln(2)\Gamma(m)}\sum\limits_{r=0}^{j-1}(-1)^r\sum\limits_{\mathcal{U}_{n,r}}\sum\limits_{\mathcal{S}_{m,r}}\frac{\prod\limits_{t=1}^{N-j+r} \lambda_{u_t}^{s_t}}{\prod\limits_{t=1}^{N-j+r} s_t!}\underbrace{\int\limits_0^\infty{\ln\left(1+\bar{\gamma}x\right)x^{\left(m-1+\sum\limits_{t=1}^{N-j+r}s_t\right)}\e^{-\left(\lambda_n+\sum\limits_{t=1}^{N-j+r}\lambda_{u_t}\right)x}\dd x}}_{I_1}.
\label{eq:Order_SNR_per_user_Ergodic_capacity_step2}
\end{equation} 
Integral $I_1$ is in the form of \cite[eq. 2.6.23.4]{prudnikov1986integrals}, hence
\begin{equation}
I_1=\frac{1}{\bar{\gamma}^{\alpha}}\frac{\pi}{\alpha\sin(\alpha\pi)}\, { }_1F_1\left(\begin{aligned} &\alpha \\ &\alpha+1\end{aligned};\frac{\zeta}{\bar{\gamma}}\right)-\Gamma(\alpha)\left(\zeta\right)^{-\alpha}\left[\ln\frac{\zeta}{\bar{\gamma}}-\psi(\alpha)-\frac{\zeta}{\bar{\gamma}(1-\alpha)}\, {}_2F_2\left(\begin{aligned}&1,1\\ &2,2-\alpha\end{aligned};\frac{\zeta}{\bar{\gamma}}\right)\right],
\label{eq:Int_I1_Nakagami_SNR}
\end{equation}
where $\alpha\definedas m+\!\!\sum\limits_{t=1}^{N-j+r}\!s_t$, $\zeta=\lambda_n+\!\!\sum\limits_{t=1}^{N-j+r}\!\lambda_{u_t}$. Substituting (\ref{eq:Int_I1_Nakagami_SNR}) in (\ref{eq:Order_SNR_per_user_Ergodic_capacity_step2}), the ergodic per-user rate can be written in closed-form for any SNR. However, RF EH usually targets short-range application scenarios. Thus, we provide approximate results for the high SNR regime, i.e., high  $\bar{\gamma}$. In this case, the argument $\frac{\zeta}{\bar{\gamma}}$ in the hypergeometric functions in $I_1$ tends to zero, and consequently both hypergeometric functions tend to 1. Moreover, the terms containing $\frac{1}{\bar{\gamma}}$ tend to zero, and $I_1$ reduces to $I_1=\Gamma(\alpha)(\zeta)^{-\alpha}\left(\psi(\alpha)+\ln\frac{\bar{\gamma}}{\zeta}\right)$. This result can be also obtained after approximating the term $\ln\left(1+\bar{\gamma}x\right)$ of $I_1$ in (\ref{eq:Order_SNR_per_user_Ergodic_capacity_step2}) as $\ln\left(\bar{\gamma}x\right)$ and using \cite[eq. 2.3.3.1]{prudnikov1986integrals} and \cite[eq. 2.6.21.2]{prudnikov1986integrals}. Hence, the resulting rate is a lower bound for the achievable rate which becomes tight for high SNR. Thus, at high SNR, the ergodic per-user rate for order-based SNR scheduling over Nakagami-$m$ fading channels reduces to the expression provided in Table \ref{tab:avg_Capacity_SNR}. 

Similarly, using (\ref{eq:Order_SNR_per_user_EH}), (\ref{eq:prod_cdfs_Nakagami}), (\ref{eq:prod_1_cdfs_Nakagami}), and the set definitions in (\ref{eq:set_Unr}) and (\ref{eq:set_Sm}), the per-user average harvested energy reduces to $\E\left[EH_{j,\U_n}\right]=$\vspace{-0.2cm}
\begin{equation}
\eta P\left(\Omega_n-\frac{\lambda_n^m}{\Gamma(m)}\sum\limits_{r=0}^{j-1}(-1)^r\sum\limits_{\mathcal{U}_{n,r}}\sum\limits_{\mathcal{S}_{m,r}}\frac{\prod\limits_{t=1}^{N-j+r} \lambda_{u_t}^{s_t}}{\prod\limits_{t=1}^{N-j+r} s_t!}\underbrace{\int\limits_0^\infty x^{\left(m+\sum\limits_{t=1}^{N-j+r} s_t\right)}\e^{-\left(\lambda_n+\sum\limits_{t=1}^{N-j+r}\lambda_{u_t}\right)x} \dd x}_{I_2}\right),
\label{eq:EH_peruser_SNR_step2_Nakagami}
\end{equation}
where  $I_2$ can be simplified using \cite[eq. 2.3.3.1]{prudnikov1986integrals} as
\begin{equation}
I_2=\Gamma\left(m+1+\sum\limits_{t=1}^{N-j+r} s_t\right) \left(\lambda_n+\sum\limits_{t=1}^{N-j+r}\lambda_{u_t}\right)^{-\left(m+1+\sum\limits_{t=1}^{N-j+r} s_t\right)}
\label{eq:I2_Nakagami_SNR}
\end{equation}
and the average per-user harvested energy reduces to the expression provided in Table \ref{tab:avg_EH_SNR}.
\subsection{Weibull Fading}
Using the pdf and the cdf of the channel power gain for Weibull fading in Table \ref{tab:pdfs_cdfs_hn},  the product term $\prod\limits_{l=1}^{j-1}{F_{h_{i_l}}(x)}f_{h_n}(x)\prod\limits_{l=j}^{N-1}{\left(1-F_{h_{i_l}}(x)\right)}$ which appears in the per-user ergodic rate and average harvested energy in (\ref{eq:Order_SNR_per_user_Ergodic_capacity}) and (\ref{eq:Order_SNR_per_user_EH}) can be written as sums of products given by
\begin{equation}
\prod\limits_{l=1}^{j-1}{F_{h_{i_l}}(x)}f_{h_n}(x)\prod\limits_{l=j}^{N-1}{\left(1-F_{h_{i_l}}(x)\right)}=k\lambda_n^k \sum\limits_{r=0}^{j-1}(-1)^r \sum\limits_{(c_1,\ldots,c_r)\in\mathcal{C}_{r,j}} x^{k-1} \e^{-\left(\sum\limits_{t=1}^{r}\lambda_{c_t}^k + \lambda_n^k+ \sum\limits_{l=j}^{N-1}\lambda_{i_l}^k\right)x^k}.
\label{eq:prod_cdfs_Weibull}
\end{equation}
Using the set definition in (\ref{eq:set_Unr}), the ergodic rate of user $n$ in  (\ref{eq:Order_SNR_per_user_Ergodic_capacity}) reduces to
\begin{equation}
\E\left[C_{j,\U_n}\right]=\frac{k\lambda_n^k}{\ln(2)}\sum\limits_{r=0}^{j-1}(-1)^r\sum\limits_{\mathcal{U}_{n,r}}\underbrace{\int\limits_0^\infty{\ln\left(1+\bar{\gamma}x\right)x^{k-1}\e^{-\left(\lambda_n^k+\sum\limits_{t=1}^{N-j+r}\lambda_{u_t}^k\right)x^k}\dd x}}_{I_3}.
\label{eq:Cap_Un_Weibull_SNR_step}
\end{equation}
Integral $I_3$ can be obtained in closed-form at high SNR by approximating $\ln(1+\bar{\gamma}x)$ by $\ln(\bar{\gamma}x)$. Using \cite[eq. 2.3.18.2]{prudnikov1986integrals} and \cite[eq. 2.6.21.2]{prudnikov1986integrals}, $I_3$ reduces to
\begin{equation}
I_3\approx\frac{1}{k\left(\lambda_n^k+\sum\limits_{t=1}^{N-j+r}\lambda_{u_t}^k\right)}\left(\ln\left(\bar{\gamma}\right)-\frac{1}{k}\left(\ln\left(\lambda_n^k+\sum\limits_{t=1}^{N-j+r}\lambda_{u_t}^k\right)+\textrm{C}\right)\right)
\label{eq:}
\end{equation}
and the average per-user rate reduces to the expression provided in Table \ref{tab:avg_Capacity_SNR}.

Using (\ref{eq:Order_SNR_per_user_EH}), (\ref{eq:prod_cdfs_Weibull}), and the set definition in (\ref{eq:set_Unr}), the average harvested energy of user $n$ can be obtained as \vspace{-0.5cm}
\begin{equation}
\E\left[EH_{j,\U_n}\right]=\eta P\Bigg(\Omega_n-k\lambda_n^k\sum\limits_{r=0}^{j-1}(-1)^r\sum\limits_{\mathcal{U}_{n,r}}\underbrace{\int\limits_0^\infty x^{k}\e^{-\left(\lambda_n^k+\sum\limits_{t=1}^{N-j+r}\lambda_{u_t}^k\right)x^k} \dd x}_{I_4}\Bigg),
\label{eq:EH_peruser_SNR_step2_Weibull}
\end{equation}
where $I_4$ can be solved using  \cite[eq. 2.3.18.2]{prudnikov1986integrals} as $I_4=\left(\lambda_n^k+\sum\limits_{t=1}^{N-j+r}\lambda_{u_t}^k\right)^{-\frac{k+1}{k}}\Gamma\left(\frac{k+1}{k}\right)$ and the average per-user harvested energy reduces to the expression in Table \ref{tab:avg_EH_SNR}.

\section{Analysis of order-based N-SNR scheduling}
\label{app:proof_EH_C_NSNR}
\subsection{Nakagami-$m$ Fading}
\label{app:EH_C_Nakagami_proof} 
Using  (\ref{eq:pdf_order_statistics}), (\ref{eq:Ergodic_capacity_NSNR}), the pdf and cdf of the normalized channel power gain from Table \ref{tab:pdfs_cdfs_hn} with $\Omega_n=1$, and the binomial theorem for $\left(F_{X}(x)\right)^{j-1}$, the ergodic rate of user $n$ reduces to
  \begin{equation}
\E[C_{j,\text{U}_n}]=\frac{m^{m}}{\Gamma(m)}\binom{N-1}{j-1}\sum\limits_{\kappa=0}^{j-1} (-1)^\kappa \binom{j-1}{\kappa}  \int\limits_0^\infty\log_2\left(1+\bar{\gamma}_nx\right) x^{m-1} \e^{-m x} \left(\frac{\Gamma(m,mx)}{\Gamma(m)}\right)^{N-j+\kappa} \dd x.
\label{eq:step_1_capacity_Nakagami}
\end{equation}
By replacing the exponent $N-j+\kappa$ by $l$, the summation variable $l$ goes from $N-j$ to $N-1$, and $\binom{j-1}{\kappa}=\binom{j-1}{l-N+j}=\binom{j-1}{N-l-1}$. So, the average rate can be rewritten as $\E[C_{j,\text{U}_n}]=$
\begin{equation}
\frac{m^{m}}{\Gamma(m)}\binom{N-1}{j-1}\sum\limits_{l=N-j}^{N-1} (-1)^{l-N+j} \binom{j-1}{N-l-1} \underbrace{\int\limits_0^\infty\log_2\left(1+\bar{\gamma}_nx\right) x^{m-1} \e^{-m x} \left(\frac{\Gamma(m,mx)}{\Gamma(m)}\right)^{l} \dd x}_{I_5}.
\label{eq:step_2_capacity_Nakagami}
\end{equation}
For integer values of fading parameter $m$, the upper normalized incomplete Gamma function $\frac{\Gamma(m,mx)}{\Gamma(m)}$ can be replaced by its series representation given by $\frac{\Gamma(m,mx)}{\Gamma(m)}=\e^{-mx}\sum\limits_{s=0}^{m-1}\frac{m^s x^s}{s!}$. Using the multinomial theorem, $\left(\sum\limits_{s=0}^{m-1}\frac{m^s x^s}{s!}\right)^l$ can be replaced by $l!\sum\limits_{\mathcal{I}_{m,l}}\left(\prod\limits_{s=0}^{m-1}\frac{\left(\frac{m^sx^s}{s!}\right)^{i_s}}{i_s!}\right)$, where we define $\mathcal{I}_{m,l}$ as 
\begin{equation}
\mathcal{I}_{m,l}=\left\{(i_0,\ldots,i_{m-1})\mid i_s\in\{0,\ldots,l\},\forall s\in\{0,\ldots, m-1\},\text{ and } \sum_{s=0}^{m-1}i_s=l\right\}.
\label{eq:I_m_l}
\end{equation} 
Thus, integral $I_5$ reduces to
\begin{equation}
I_5=l!\sum\limits_{\mathcal{I}_{m,l}}\left(\prod\limits_{s=0}^{m-1}\frac{\left(\frac{m^s}{s!}\right)^{i_s}}{i_s!}\right)\underbrace{\int\limits_0^\infty\log_2\left(1+\bar{\gamma}_nx\right) \,\e^{-m(1+l) x} \,x^{m-1+\sum\limits_{s=0}^{m-1}s\,i_s}\dd x}_{I_6}.
\label{eq:I_5}
\end{equation}
Integral $I_6$ has the same form as $I_1$ in (\ref{eq:Order_SNR_per_user_Ergodic_capacity_step2}). Thus, a similar analysis can be applied resulting in exact as well as lower bound expressions for the per-user ergodic rate of order-based SNR scheduling over Nakagami-$m$ fading. 

The average per-user harvested energy can be derived from (\ref{eq:EH_Un_NSNR}), where similar to (\ref{eq:step_2_capacity_Nakagami}), $\frac{1}{N}\E[X_{(j)}]$ can be written as
\begin{equation}
\frac{1}{N}\E[X_{(j)}]=\frac{m^{m}}{\Gamma(m)}\binom{N-1}{j-1}\sum\limits_{l=N-j}^{N-1} (-1)^{l-N+j} \binom{j-1}{N-l-1} \underbrace{\int\limits_0^\infty x^{m} \e^{-m x} \left(\e^{-mx}\sum\limits_{s=0}^{m-1}\frac{m^s x^s}{s!}\right)^{l} \dd x}_{I_7}.
\label{eq:EH_UI_NSNR_IND_Nakagami_step1}
\end{equation}
Applying the multinomial theorem to $\left(\sum\limits_{s=0}^{m-1}\frac{m^s x^s}{s!}\right)^{l}$ for integer $m$, $I_7$ reduces to
\begin{equation}
I_7= l!\sum\limits_{\mathcal{I}_{m,l}}\left(\prod\limits_{s=0}^{m-1}\frac{\left(\frac{m^s}{s!}\right)^{i_s}}{i_s!}\right)\underbrace{\int\limits_{0}^{\infty}x^{m+\sum\limits_{s=0}^{m-1}s\,i_s} \e^{-m(1+l)x} \dd x}_{I_8}.
\end{equation}
Using \cite[eq. 2.3.3.1]{prudnikov1986integrals}, $I_8$ reduces to
\begin{equation}
I_8=\left(\frac{1}{m(1+l)}\right)^{m+1}\left(\prod\limits_{s=0}^{m-1}\left(\frac{1}{m(1+l)}\right)^{s\,i_s}\right) \Gamma\left(m+1+\sum\limits_{s=0}^{m-1}s i_s\right).
\label{eq:EH_Nakagami_I8}
\end{equation}
Combining  (\ref{eq:EH_Un_NSNR}) and (\ref{eq:EH_UI_NSNR_IND_Nakagami_step1})-(\ref{eq:EH_Nakagami_I8}), the average per-user harvested energy for order-based N-SNR scheduling over Nakagami-$m$ fading reduces to the expression in Table \ref{tab:avg_EH_NSNR}.
\subsection{Weibull Fading}
Using  (\ref{eq:pdf_order_statistics}), (\ref{eq:Ergodic_capacity_NSNR}), and  the pdf and cdf of the normalized channel power gain from Table \ref{tab:pdfs_cdfs_hn} with $\Omega_n=1$, the ergodic rate  of user $n$ is given by
\begin{equation}
\E\left[C_{j,\text{U}_n}\right]=\frac{1}{\ln(2)}\binom{N-1}{j-1}\underbrace{\int\limits_0^\infty\ln(\bar{\gamma}_nx)f_X(x) F_X(x)^{j-1}(1-F_X(x))^{N-j} \dd x}_{I_9},
\label{eq:weibull_C_Un_start}
\end{equation}
where at high SNR $\ln(1+\bar{\gamma}_nx)\approx\ln(\bar{\gamma}_nx)$ is used, hence the resulting rate is a lower bound for the achievable rate which becomes tight at high SNR. Substituting $1-F_X(x)$ by $u$, i.e., $x=\frac{1}{\Gamma\left(1+\frac{1}{k}\right)}\left(\ln\left(\frac{1}{u}\right)\right)^{\frac{1}{k}}$ and using the binomial theorem for $(1-u)^{j-1}$, $I_9$ reduces to
\begin{equation}
\begin{aligned}
I_9&=\sum\limits_{l=0}^{j-1}(-1)^l\binom{j-1}{l}\int\limits_0^1\ln\left(\frac{\bar{\gamma}_n}{\Gamma\left(1+\frac{1}{k}\right)}\left(\ln\left(\frac{1}{u}\right)\right)^{\frac{1}{k}}\right) u^{N-j+l} \dd u\\
&=\sum\limits_{l=0}^{j-1}(-1)^l\binom{j-1}{l}\frac{1}{N-j+l+1}\left(\ln\left(\frac{\bar{\gamma}_n}{\Gamma\left(1+\frac{1}{k}\right)}\right)-\frac{1}{k}\left(\text{C}+\ln(N-j+l+1)\right)\right),
\end{aligned}
\label{eq:weibull_NSNR_C_Un_step}
\end{equation}
where we used the integral in \cite[eq. 4.325.8]{jeffrey2007table}.
Using (\ref{eq:weibull_C_Un_start}) and (\ref{eq:weibull_NSNR_C_Un_step}), the average per-user rate for  order-based N-SNR scheduling over Weibull fading reduces after simple manipulations to the expression in Table \ref{tab:avg_Capacity_NSNR}.

In order to compute the average per-user harvested energy in (\ref{eq:EH_Un_NSNR}), we use (\ref{eq:pdf_order_statistics}) to write $
\frac{1}{N}\E\left[X_{(j)}\right]=\binom{N-1}{j-1}\int\limits_0^\infty x f_X(x) F_X(x)^{j-1}(1-F_X(x))^{N-j} \dd x$.
Using again the substitution $u=1-F_X(x)$ and the binomial theorem, $\frac{1}{N}\E\left[X_{(j)}\right]$ reduces to 
\begin{equation}
\begin{aligned}
\frac{1}{N}\E\left[X_{(j)}\right]&=\binom{N-1}{j-1}\frac{1}{\Gamma\left(1+\frac{1}{k}\right)}\sum\limits_{l=0}^{j-1}(-1)^l\binom{j-1}{l}\int\limits_0^1\left(\ln\left(\frac{1}{u}\right)\right)^{\frac{1}{k}}u^{N-j+l} \dd u\\
&=\binom{N-1}{j-1}\sum\limits_{l=0}^{j-1}(-1)^l\binom{j-1}{l}\left(N-j+l+1\right)^{-\left(1+\frac{1}{k}\right)},
\label{eq:weibull_EH_Un_step2}
\end{aligned}
\end{equation}
where the integral was solved using \cite[eq. 4.272.6]{jeffrey2007table}. Combining (\ref{eq:weibull_EH_Un_step2}) and (\ref{eq:EH_Un_NSNR}), the average per-user harvested energy for order-based N-SNR scheduling over Weibull fading given in Table \ref{tab:avg_EH_NSNR} is obtained. 

\section{Feasibility Conditions for order-based ET Scheduling}
\label{app:feasibility_conditions_proof}
Order-based ET scheduling may fail to provide all users with ET for one of the following reasons:
\begin{enumerate}
	\item Some user $n$ is required to be scheduled more often than possible. That is $\exists n: p_n>\frac{|\Sa|}{N} \overset{\text{from (\ref{eq:pn_pn_conditioned})}}{\equiv} \pr(\U_n|O_n\in\Sa)>1$. Thus, the first feasibility condition follows.
	\item For certain combinations of users in $\Sa$, the sum of the required probabilities that one of them accesses the channel exceeds one. That is, $\exists$ a combination $(k_1,\ldots,k_{|\Sa|})$ of $\{1,\ldots,N\}$, for which $\sum_{l=1}^{|\Sa|}\pr(\U_{k_l}|\Sa=\{O_{k_1},\ldots,O_{k_{|\Sa|}}\})>1$.
\end{enumerate}
To find a simple condition for the second case, we first synthesize $\text{Pr}(\U_n|O_n\in\Sa)$ using the law of total probability as 
\begin{equation}
\text{Pr}(\U_n|O_n\!\in\!\Sa)\!=\!\frac{1}{\binom{N-1}{|\Sa|-1}}\sum\limits_{\mathcal{C}'_n}\pr(\U_n|\Sa\!=\!\{O_n,O_{i_1}\!,\!\ldots\!,O_{i_{|\Sa|-1}\!}\}), 
\end{equation}
where $\mathcal{C}'_n$ is the set of all $\binom{N-1}{|\Sa|-1}$ combinations $(i_1,\ldots,i_{|\Sa|-1})$ from $\{1,\ldots\!,n\!-\!1\!,n\!+\!1,N\!\}$. Thus, from (\ref{eq:pn_pn_conditioned}),
\begin{equation}
\binom{N}{|\Sa|}p_n=\sum\limits_{\mathcal{C}'_n}\pr(\U_n|\Sa=\{O_n,O_{i_1},\ldots,O_{i_{|\Sa|-1}}\})
\label{eq:p_n_synthesized}
\end{equation}
holds for all $n\in\{1,\ldots,N\}$. In order to check that 
\begin{equation}
\sum\limits_{l=1}^{|\Sa|}\pr(\U_{k_l}|\Sa=\{O_{k_1},\ldots,O_{k_{|\Sa|}}\})=1
\label{eq:conditional_prob_must_1}
\end{equation}
holds for all combinations $(k_1,\ldots,k_{|\Sa|})$  drawn from  $\{1,\ldots,N\}$, we observe that adding $|\Sa|$ equations of (\ref{eq:p_n_synthesized}) for users of indices $n=k_1,\ldots,k_{|\Sa|}$ results in 
\begin{equation}\binom{N}{|\Sa|}\sum\limits_{l=1}^{|\Sa|}p_{k_l}=\sum\limits_{l=1}^{|\Sa|}\pr(\U_{k_l}|\Sa=\{O_{k_1},\ldots,O_{k_{|\Sa|}}\})+\ldots.
\end{equation}
Hence, applying (\ref{eq:conditional_prob_must_1}) and limiting every remaining probability term to $1$, the whole summation is limited to 
\begin{equation}
\binom{N}{|\Sa|}\sum\limits_{l=1}^{|\Sa|} p_{k_{l}} \leq \hspace{-0.3cm}\underbrace{\binom{N-1}{|\Sa|-1}}_{\begin{aligned}&\text{total number of}\\[-0.5em] & \text{probabilities per}\\[-0.5em] & \text{eq. in (\ref{eq:p_n_synthesized}})\end{aligned}}\hspace{0.1cm}\underbrace{|\Sa|}_{\begin{aligned}&\text{number of} \\[-0.5em] &\text{eqs. added}\end{aligned}}+\underbrace{(1-|\Sa|)}_{\begin{aligned}&\text{replacing the number} \\[-0.5em] &\text{of probability terms}\\[-0.5em] &\text{that add up to 1 by 1}\end{aligned}}.
\label{eq:feasibility_condition_Sa}
\end{equation}
Moreover, adding $L>|\Sa|$ equations of (\ref{eq:p_n_synthesized}) for the users with indices $k_1,\ldots,k_L$, and applying (\ref{eq:conditional_prob_must_1}) for every $|\Sa|$-length combination from the set $\{k_1,\ldots,k_L\}$, the whole summation is limited to
\begin{equation}
\binom{N}{|\Sa|}\sum\limits_{l=1}^{L} p_{k_{l}} \leq \binom{N-1}{|\Sa|-1}L+\binom{L}{|\Sa|}(1-|\Sa|),
\end{equation}
which must hold for every combination $(k_1,\ldots,k_L)$ in $\{1,\ldots,N\}$ and for every $L=|\Sa|,\ldots,N$. Hence, the second feasibility condition follows.

\bibliographystyle{IEEEtran}
\bibliography{references}

\begin{thebibliography}{10}
\providecommand{\url}[1]{#1}
\csname url@samestyle\endcsname
\providecommand{\newblock}{\relax}
\providecommand{\bibinfo}[2]{#2}
\providecommand{\BIBentrySTDinterwordspacing}{\spaceskip=0pt\relax}
\providecommand{\BIBentryALTinterwordstretchfactor}{4}
\providecommand{\BIBentryALTinterwordspacing}{\spaceskip=\fontdimen2\font plus
\BIBentryALTinterwordstretchfactor\fontdimen3\font minus
  \fontdimen4\font\relax}
\providecommand{\BIBforeignlanguage}[2]{{%
\expandafter\ifx\csname l@#1\endcsname\relax
\typeout{** WARNING: IEEEtran.bst: No hyphenation pattern has been}%
\typeout{** loaded for the language `#1'. Using the pattern for}%
\typeout{** the default language instead.}%
\else
\language=\csname l@#1\endcsname
\fi
#2}}
\providecommand{\BIBdecl}{\relax}
\BIBdecl

\bibitem{Multiuser_scheduling_Morsi}
R.~{Morsi}, D.~S. {Michalopoulos}, and R.~{Schober}, ``{Multi-user Scheduling
  Schemes for Simultaneous Wireless Information and Power Transfer},'' in
  \emph{IEEE Intern. Conf. Commun. (ICC)}, June 2014, pp. 5005--5010.

\bibitem{powercast}
\BIBentryALTinterwordspacing
{Powercast Coporation}, ``{RF Energy Harvesting and Wireless Power for
  Low-Power Applications},'' 2011. [Online]. Available:
  \url{http://www.mouser.com/pdfdocs/Powercast-Overview-2011-01-25.pdf}
\BIBentrySTDinterwordspacing

\bibitem{RFtoDC2008}
T.~Paing, J.~Shin, R.~Zane, and Z.~Popovic, ``{Resistor Emulation Approach to
  Low-Power RF Energy Harvesting},'' \emph{IEEE Trans. Power Electronics},
  vol.~23, pp. 1494--1501, May 2008.

\bibitem{Varshney2008}
L.~Varshney, ``{Transporting Information and Energy Simultaneously},'' in
  \emph{IEEE Intern. Symp. Inform. Theory (ISIT)}, July 2008, pp. 1612--1616.

\bibitem{Shannon_meets_tesla_Grover2010}
P.~Grover and A.~Sahai, ``{Shannon Meets Tesla: Wireless Information and Power
  Transfer},'' in \emph{IEEE Intern. Symp. Inform. Theory (ISIT)}, June 2010,
  pp. 2363--2367.

\bibitem{MIMO_Broadcasting_Zhang2011_Journal}
R.~Zhang and C.~K. Ho, ``{MIMO Broadcasting for Simultaneous Wireless
  Information and Power Transfer},'' \emph{IEEE Trans. on Wireless Commun.},
  vol.~12, no.~5, pp. 1989--2001, May 2013.

\bibitem{WIPT_Architecture_Rui_Zhang_2012}
X.~Zhou, R.~Zhang, and C.~K. Ho, ``{Wireless Information and Power Transfer:
  Architecture Design and Rate-Energy Tradeoff},'' \emph{IEEE Trans. on
  Commun.}, vol.~61, no.~11, pp. 4754--4767, November 2013.

\bibitem{EH_ID_Scenarios_Clerckx}
J.~Park and B.~Clerckx, ``{Joint Wireless Information and Energy Transfer in a
  Two-User MIMO Interference Channel},'' \emph{IEEE Trans. on Wireless
  Commun.}, vol.~12, no.~8, pp. 4210--4221, August 2013.

\bibitem{Multiuser_MISO_beamforming2013}
J.~Xu, L.~Liu, and R.~Zhang, ``{Multiuser MISO Beamforming for Simultaneous
  Wireless Information and Power transfer},'' in \emph{IEEE Intern. Conf. on
  Acoustics, Speech and Signal Processing (ICASSP),}, May 2013, pp. 4754--4758.

\bibitem{Multiuser_OFDM_Kwan_Schober_Journal}
D.~Ng, E.~Lo, and R.~Schober, ``{Wireless Information and Power Transfer:
  Energy Efficiency Optimization in OFDMA Systems},'' \emph{IEEE Trans. on
  Wireless Commun.}, vol.~12, no.~12, pp. 6352--6370, December 2013.

\bibitem{Throughput_Maximization_for_WPCN_Zhang2013}
H.~Ju and R.~Zhang, ``{Throughput Maximization in Wireless Powered
  Communication Networks},'' \emph{IEEE Trans. on Wireless Commun.}, vol.~13,
  no.~1, pp. 418--428, January 2014.

\bibitem{Performance_Analysis_MUD_Alouini_2004}
L.~Yang and M.-S. Alouini, ``{Performance Analysis of Multiuser Selection
  Diversity},'' in \emph{IEEE Intern. Conf. Commun. (ICC)}, vol.~5, June 2004,
  pp. 3066--3070.

\bibitem{Unified_Scheduling_approach}
X.~Wang, G.~Giannakis, and A.~Marques, ``{A Unified Approach to QoS-Guaranteed
  Scheduling for Channel-Adaptive Wireless Networks},'' \emph{Proceedings of
  the IEEE}, vol.~95, no.~12, pp. 2410--2431, 2007.

\bibitem{equal_throughput_2009}
A.~Fernekess, A.~Klein, B.~Wegmann, and K.~Dietrich, ``{Analysis of Cellular
  Mobile Networks using Fair Throughput Scheduling},'' in \emph{IEEE 20th
  Intern. Symp. on Personal, Indoor and Mobile Radio Commun.}, Sept. 2009, pp.
  2945 --2949.

\bibitem{MeijerG1990}
V.~S. Adamchik and O.~I. Marichev, ``{The Algorithm for Calculating Integrals
  of Hypergeometric Type Functions and its Realization in REDUCE System},'' in
  \emph{Intern. Symp. on Symbolic and Algebraic Comput.}, ser. ISSAC '90.\hskip
  1em plus 0.5em minus 0.4em\relax New York, NY, USA: ACM, 1990, pp. 212--224.

\bibitem{prudnikov1986integrals}
A.~Prudnikov, I.~Brychkov, and O.~Mari{\v{c}}ev, \emph{Integrals and Series:
  Elementary Functions}, ser. Integrals and Series.\hskip 1em plus 0.5em minus
  0.4em\relax New York : Gordon and Breach Science Publishers, 1986, vol.~1.

\bibitem{simon2005digital}
M.~Simon and M.~Alouini, \emph{Digital Communication over Fading Channels},
  ser. Wiley Series in Telecommunications and Signal Processing.\hskip 1em plus
  0.5em minus 0.4em\relax Wiley, 2005.

\bibitem{WBAN_Weibull_Justification}
D.~Smith, J.~Zhang, L.~Hanlen, D.~Miniutti, D.~Rodda, and B.~Gilbert, ``{A
  Simulator for the Dynamic On-body Area Propagation Channel},'' in \emph{IEEE
  Intern. Symp. Antennas and Propagation Society}, 2009, pp. 1--4.

\bibitem{EXP_Approx_Marcum_Q}
M.~Bocus, C.~Dettmann, and J.~Coon, ``{An Approximation of the First Order
  Marcum Q-Function with Application to Network Connectivity Analysis},''
  \emph{IEEE Commun. Letters}, vol.~17, no.~3, pp. 499--502, March 2013.

\bibitem{Shannon_Capacity_fading_channels2005}
N.~Sagias, G.~Tombras, and G.~Karagiannidis, ``{New Results for the Shannon
  Channel Capacity in Generalized Fading Channels},'' \emph{IEEE Commun.
  Letters}, vol.~9, no.~2, pp. 97 -- 99, Feb. 2005.

\bibitem{order_statistics_ind_1994}
N.~Balakrishnan, ``{Order Statistics From Non-identical Exponential Random
  Variables and Some Applications},'' \emph{Comput. Stat. Data Anal.}, vol.~18,
  no.~2, pp. 203--253, 1994.

\bibitem{Order_Statistics_David_Nagaraja}
H.~A. David and H.~N. Nagaraja, \emph{{Order Statistics}}, 3rd~ed.\hskip 1em
  plus 0.5em minus 0.4em\relax John Wiley \& Sons, Inc., August 2003.

\bibitem{PathLoss_indoor_Rappaport_1992}
S.~Seidel and T.~Rappaport, ``{914 MHz Path Loss Prediction Models for Indoor
  Wireless Communications in Multifloored Buildings},'' \emph{IEEE Trans. on
  Antennas and Propagation}, vol.~40, no.~2, pp. 207--217, 1992.

\bibitem{wang2009electronic}
L.~Wang, Y.~Chang, and K.~Cheng, \emph{Electronic Design Automation: Synthesis,
  Verification, and Test}, ser. Systems on Silicon.\hskip 1em plus 0.5em minus
  0.4em\relax Elsevier Science, 2009.

\bibitem{jeffrey2007table}
A.~Jeffrey and D.~Zwillinger, \emph{Table of Integrals, Series, and Products},
  7th~ed., ser. Table of Integrals, Series, and Products Series.\hskip 1em plus
  0.5em minus 0.4em\relax Elsevier Science, 2007.

\end{thebibliography}
\end{document}